\documentclass[aoas]{imsart}

\RequirePackage{amsthm,amsmath,amsfonts,amssymb}
\RequirePackage{bbm}
\RequirePackage[authoryear]{natbib}
\RequirePackage[colorlinks,citecolor=blue,urlcolor=blue]{hyperref}
\RequirePackage{graphicx}
\RequirePackage{scrextend}
\RequirePackage[ruled,norelsize,vlined]{algorithm2e}
\RequirePackage{colortbl} 
\RequirePackage{xcolor}
\RequirePackage{booktabs}
\RequirePackage{bbm}
\RequirePackage{caption}
\RequirePackage{subcaption}

\changefontsizes{13.6pt}

\startlocaldefs

\newcommand{\bb}{ {\bf b} }

\newcommand{\E}{\mathbb{E}}

\newcommand{\bH}{ {\bf H} }

\newcommand{\bef}{ {\bf f} }
\newcommand{\bm}{ {\bf m} }

\newcommand{\bZ}{ {\bf Z} }

\newcommand{\bC}{ {\bf C} }

\newcommand{\bD}{ {\bf D} }

\newcommand{\bQ}{ {\bf Q} }
\newcommand{\bV}{ {\bf V} }

\newcommand{\bg}{ {\bf g} }
\newcommand{\bW}{ {\bf W} }
\newcommand{\bL}{ {\bf L} }

\newcommand{\bT}{ {\bf T} }

\newcommand{\br}{ {\bf r} }

\newcommand{\bgamma}{ {\boldsymbol \gamma} }
\newcommand{\bepsilon}{ {\boldsymbol \epsilon} }
\newcommand{\bvarepsilon}{ {\boldsymbol \varepsilon} }

\newcommand{\bDelta}{ {\boldsymbol \Delta} }
\newcommand{\bbeta}{ {\boldsymbol \beta} }
\newcommand{\bmu}{ {\boldsymbol \mu} }
\newcommand{\bSigma}{ {\boldsymbol \Sigma} }
\newcommand{\bsigma}{ {\boldsymbol \sigma} }
\newcommand{\bOmega}{ {\boldsymbol \Omega} }
\newcommand{\bomega}{ {\boldsymbol \omega} }

\newcommand{\boeta}{ {\boldsymbol \eta} }

\newcommand{\bnu}{ {\boldsymbol \nu} }

\newcommand\scalemath[2]{\scalebox{#1}{\mbox{\ensuremath{\displaystyle #2}}}}

\theoremstyle{plain}

\newtheorem{theorem}{Theorem}[section]

\newtheorem{Proposition}[theorem]{Proposition}

\theoremstyle{remark}

\endlocaldefs

\begin{document}

\begin{frontmatter}
\title{\small Learning and forecasting of age--specific period mortality via B--spline processes with locally--adaptive dynamic coefficients}
\runtitle{Learning and forecasting of age--specific period mortality}

\begin{aug}
%%%%%%%%%%%%%%%%%%%%%%%%%%%%%%%%%%%%%%%%%%%%%%%
%%%%%%%%%%%%%%%%%%%%%%%%%%%%%%%%%%%%%%%%%%%%%%%
\author[A]{\fnms{Federico} \snm{Pavone}\ead[label=e1,mark]{federico.pavone@phd.unibocconi.it}},
\author[B]{\fnms{Sirio} \snm{Legramanti}\ead[label=e2,mark]{sirio.legramanti@unibg.it}}
\and
\author[A]{\fnms{Daniele} \snm{Durante}\ead[label=e3,mark]{daniele.durante@unibocconi.it}}
%%%%%%%%%%%%%%%%%%%%%%%%%%%%%%%%%%%%%%%%%%%%%%
%%%%%%%%%%%%%%%%%%%%%%%%%%%%%%%%%%%%%%%%%%%%%%
\address[A]{Department  of Decision Sciences and Institute for Data Science and Analytics,
Bocconi University, \\
\printead{e1,e3}}

\vspace{-5pt}

\address[B]{Department of Economics, University of Bergamo,
\printead{e2}}

\end{aug}

%%%%%%%%%%%%%%%%%%%%%%%%%%%%%%%%%%%%%%%%%%%%%%%
%%%%%%%%%%%%%%%%%%%%%%%%%%%%%%%%%%%%%%%%%%%%%%%

\begin{abstract}
Although the analysis of human mortality has a well--established history, the attempt to accurately  forecast future death--rate patterns for different age groups and time horizons still attracts active research. Such a predictive focus has motivated an increasing shift towards more flexible representations of age--specific period mortality trajectories at the cost of reduced interpretability. Although this perspective has led to successful predictive strategies, the inclusion of interpretable structures in modeling  of human mortality can be, in fact, beneficial for improving forecasts. We pursue this direction via a novel  \textsc{b}--spline process with locally--adaptive dynamic coefficients. Such a process outperforms state--of--the--art forecasting strategies by explicitly  incorporating the core structures of period mortality within an interpretable formulation which enables inference  on age--specific mortality trends and the corresponding rates of change across time. This is obtained by modeling the age--specific death counts via a Poisson log--normal model parameterized through a linear combination of \textsc{b}--spline bases with dynamic coefficients that characterize time changes in mortality rates via  suitably defined stochastic differential equations. While flexible, the resulting formulation can be accurately approximated by a Gaussian state--space model that facilitates closed--form Kalman filtering, smoothing and forecasting, for both the trends of the spline coefficients and the corresponding  first derivatives, which measure rates of change in mortality  for different age groups. As illustrated in applications to mortality data from different countries, the proposed model outperforms state--of--the--art methods both in point forecasts and in calibration of predictive intervals. Moreover,  it unveils substantial differences in mortality patterns across countries and ages, both in the past decades and during the \textsc{covid}--19 pandemic.
\end{abstract}

\begin{keyword}
\kwd{\textsc{b}--spline}
\kwd{Death rate}
\kwd{Kalman filter}
\kwd{Mortality forecasting}
\kwd{Nested Gaussian process}
\end{keyword}

\end{frontmatter}

%%%%%%%%%%%%%%%%%%%%%%%%%%%%%%%%%%%%%%%%%%%%%%
%%%%%%%%%%%%%%%%%%%%%%%%%%%%%%%%%%%%%%%%%%%%%%
\section{Introduction} \label{sec_1}
Since the fundamental contribution by  \citet{lee1992modeling} on stochastic modeling and forecasting of human mortality patterns, several efforts  have been devoted towards the development of increasingly accurate strategies for predicting the future evolution of death rates for different age groups and countries \citep[e.g.,][]{booth2008mortality,currie2016fitting,hunt2021structure}. Due to its direct impact in guiding social, economic, environmental and health--care policies, such an endeavor is of paramount interest in a variety of fields, including demography \citep[e.g.,][]{lee2001evaluating,li2005coherent,raftery2013bayesian, li2013extending,camarda2019smooth}, actuarial sciences \citep[e.g.,][]{renshaw2003lee,renshaw2006cohort,cairns2006two,plat2009stochastic,currie2016fitting} and statistics \citep[e.g.,][]{lee1992modeling, dellaportas2001bayesian,hyndman2007robust,alexopoulos2019bayesian,aliverti2021dynamic}, among others. Nonetheless, despite this collective effort, there is still the lack of a consensus on a superior solution. In fact, several peculiar characteristics of age--specific period mortality trajectories keep motivating active and ongoing innovations in stochastic modeling and forecasting of death rates via increasingly flexible representations. 

\begin{figure}       
\centering
     \begin{subfigure}[b]{1\textwidth}
         \centering
         \includegraphics[width=1\textwidth]{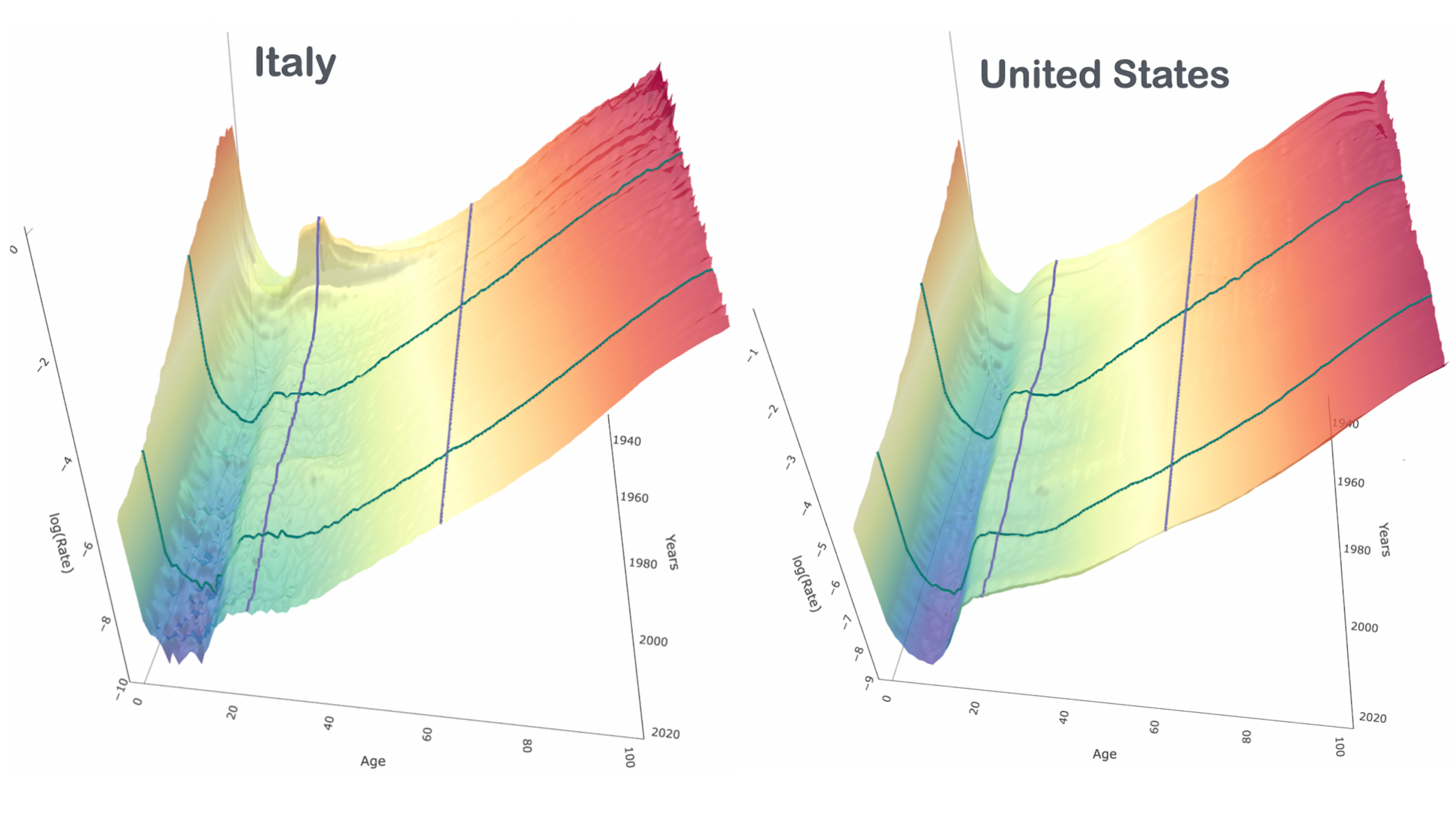}
         \hfill
     \end{subfigure}
 \begin{subfigure}[b]{1\textwidth}
         \centering
         \includegraphics[width=1\textwidth,height=0.42\textwidth]{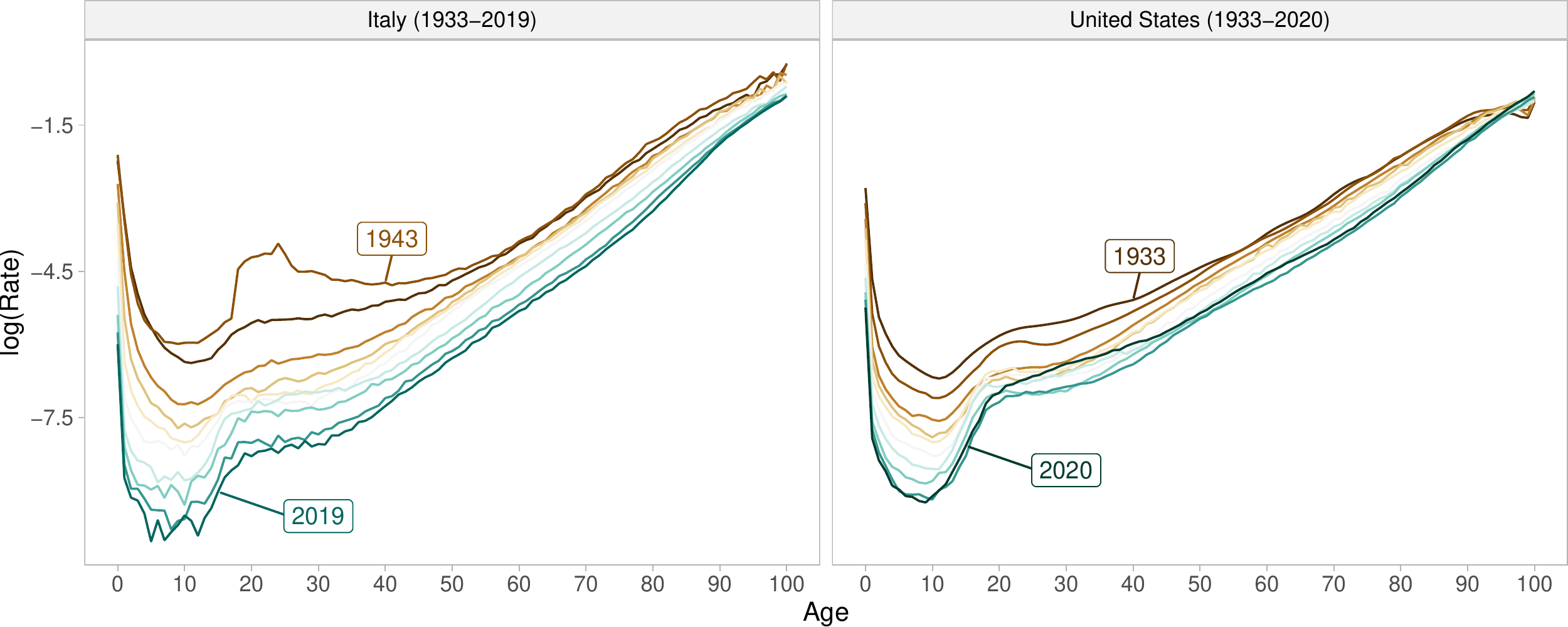}
     \end{subfigure}
         \caption{\small Graphical representations of the observed age--period log--mortality rates from 1933 until 2019 for Italy, and from 1933 to 2020 for the United States. Top panels provide 3\textsc{d} visualizations of the age--period log--mortality rate surface, whereas bottom panels comprise a 2\textsc{d} illustration of the age--specific trajectories for each period. Data are retrieved from the \texttt{Human} \texttt{Mortality} \texttt{Database} (\url{https://www.mortality.org/}). See the online article for the color version of this figure.}
                \label{fig:data}
\end{figure}

As illustrated in Figure \ref{fig:data}, the age--period mortality surfaces exhibit a combination of global and local variations. When expressed as a function of age, these mortality trajectories display similar and generally--smooth shapes, whereas the overall dynamic evolution of these trajectories across periods exhibits a progressive downward shift, whose rate of change varies locally  with both the age classes and years. Although the inclusion of these core structures is expected to enhance both inference and forecasting performance, current literature still lacks a statistical model that can effectively address such goals within a single formulation. In fact, while successful extensions of the age--period bilinear formulation by  \citet{lee1992modeling} and of the additive age--period--cohort representation in, e.g., \citet{holford1983estimation} improve the flexibility via more general basis expansions of age effects with time--varying coefficients,  the selected bases are either simple parametric functions typically active in the whole age range \citep{brouhns2002poisson,czado2005bayesian,cairns2006two,delwarde2007smoothing,plat2009stochastic,cairns2009quantitative,haberman2011comparative,o2012explaining,wong2018bayesian} or are inferred through functional principal components analysis \citep{hyndman2007robust,hyndman2013coherent}. This implies that the induced death--rate forecasts are mainly based on a combination of global trends in mortality across ages which do not explicitly account for local heterogeneities in mortality levels and the corresponding rates of change for specific age classes. Recalling Figure \ref{fig:data}, the mortality patterns exhibit both global and local variations across years and ages, thereby suggesting that a suitable representation capable of including these two  behaviors would yield improved forecasts with respect to those obtained under a mainly--global perspective.

An effective option for addressing the aforementioned goal is to rely on a more structured and interpretable basis expansion that incorporates possible heterogeneity in mortality patterns for different age groups. Within this framework, the contribution by \citet{heligman1980age} provides a first effective answer which expresses the age patterns of mortality via a combination of three basis functions corresponding to infant mortality, accident hump and elderly--age mortality; see also  \citet{dellaportas2001bayesian}, \citet{mazzuco2018mortality}  and \citet{alexopoulos2019bayesian} for subsequent extensions. While these formulations yield interpretable inference, the combination of only three  bases is generally not sufficient to flexibly characterize the broad spectrum of global and local variations in the age--specific mortality rates across years, thereby affecting forecasting performance \citep[e.g.,][]{camarda2019smooth}. To overcome this issue, a possible solution consists in specifying a richer set of basis functions, each active --- i.e., non--zero --- only in a subset of the ages, with these subsets varying across bases in order to cover the whole age range. Expressing the age pattern of mortality through a linear combination of these basis functions yields a globally--smooth, yet flexible, representation which additionally accounts for possible local heterogeneities in specific age classes via the control  on the coefficients for the bases active in those classes. Such a direction has been partially explored in \citet{currie2004smoothing} via two--dimensional penalized \textsc{b}--splines to jointly characterize age--period patterns of mortality; see also  \citet{camarda2019smooth} for a recent effective extension of this approach which incorporates suitable constraints and prior knowledge to improve forecasting performance. Although both formulations provide a sensible representation of age--period mortality surfaces, the two--dimensional \textsc{b}--splines perspective enforces a constant smoothing both across ages and periods, and prevents from treating the period component as a time--indexed stochastic process on which to impose a suitable dynamic model for principled inference and forecasting. As illustrated in Figure  \ref{fig:data} with a focus on two of the countries analyzed in our application, while the age pattern of mortality often exhibits a smooth trajectory, the time changes in such a trajectory fluctuate between periods of rapid and slow variations, affecting the age classes with different magnitudes. These peculiar characteristics necessarily require a careful statistical model which can effectively combine interpretable basis expansions for the age patterns of mortality with a flexible stochastic process having locally--varying smoothness for the dynamic evolution of such patterns across periods. While the aforementioned contributions include some of these structures, there is still the lack of a unique representation that  effectively accounts for all these peculiar characteristics within a single formulation.

Motivated by the above discussion and by the mortality data discussed in Section~\ref{sec_1.1},  we cover this  gap in Section~\ref{sec_2} by defining a Poisson log--normal model for the age--specific death counts whose rate is parameterized via a novel \textsc{b}--spline process with locally--adaptive dynamic coefficients which extends the nested Gaussian process by \citet{zhu2013locally} in a number of directions. Our novel formulation characterizes the age patterns of mortality via a suitable combination of  interpretable \textsc{b}--spline bases --- each active in  different age intervals --- and incorporates flexible dynamic changes in such patterns by allowing the splines coefficients to evolve in time via a system of stochastic differential equations that account for locally--varying smoothness in time trajectories, and facilitate borrowing of information across the coefficients of contiguous splines. Such a representation is conceptually and practically more suitable than the bivariate \textsc{b}--splines approach in \citet{currie2004smoothing} and \citet{camarda2019smooth} since it allows to properly treat the period component as a dynamic locally--adaptive stochastic process rather than just a function of time with a constant smoothness. In addition, it yields a more flexible characterization of age--period mortality patterns relative to classical parametric extensions of  the \citet{lee1992modeling} model in, e.g., \citet{brouhns2002poisson,czado2005bayesian,cairns2006two,delwarde2007smoothing,plat2009stochastic,cairns2009quantitative,haberman2011comparative,o2012explaining} and \citet{wong2018bayesian}, while preserving interpretability via the use of  \textsc{b}--spline bases instead of those inferred from, e.g.,  functional principal components \citep{hyndman2007robust,hyndman2013coherent}. 

As clarified in Sections~\ref{sec_2}--\ref{sec_4}, these advancements yield a statistical model which is both  flexible and interpretable, thereby improving accuracy in point forecasts, calibration of predictive intervals, and inference potentials relative to state--of--the--art formulations, at no additional cost in computational tractability. In fact, in Section~\ref{sec_3} we derive a provably accurate Gaussian state--space approximation of the proposed model that allows the implementation of closed--form Kalman filter updates for smoothing, filtering and forecasting of both the trends and the first derivatives for the trajectories of the spline coefficients. This computational tractability  is in contrast with recent flexible representations that require \textsc{mcmc} methods to benefit from a fully--Bayesian approach, which further allows the choice of priors for the structural model parameters \citep[e.g.,][]{wong2018bayesian,alexopoulos2019bayesian}. Moreover, unlike for state--of--the--art extensions of the \citet{lee1992modeling} model that generally employ an \textsc{arima} formulation for the dynamic parameters \citep{brouhns2002poisson,czado2005bayesian,cairns2006two,hyndman2007robust,plat2009stochastic,cairns2009quantitative,haberman2011comparative,o2012explaining,hyndman2013coherent}, our proposal explicitly incorporates and flexibly learns not only mortality trends but also the corresponding time--varying rates of change. Although the importance of accounting for the dynamic rates of change in mortality forecasting has been recently illustrated in \citet{camarda2019smooth}, such a concept has received limited attention to date, and there is a lack of statistical models which explicitly include and learn these higher--level patterns within a single formulation. The empirical performance illustrated in Section~\ref{sec_4} for our proposed model clarifies that this additional structure is not only beneficial in delivering improved point forecasts and predictive intervals than state--of--the--art competing methods, but also allows to quantify and compare relevant mortality accelerations experienced both in past and recent years across different countries and age groups. For example, our proposed model reveals substantially different patterns in age--specific mortality across countries during the last two decades and in the recent \textsc{covid}--19 pandemic. Concluding remarks and future research directions are provided in Section~\ref{sec_5}, whereas codes and tutorial implementations are available at \url{https://github.com/fpavone/BSP-mortality}.

%%%%%%%%%%%%%%%%%%%%%%%%%%%%%%%%%%%%%%%%%%%%%%
\subsection{Motivating Application}\label{sec_1.1}
The novel \textsc{b}--spline process (\textsc{bsp}) developed  in Sections~\ref{sec_2}--\ref{sec_3} is meant to provide a general modeling and forecasting framework that can be applied to any country of interest. To this end, the motivating application we consider in Section~\ref{sec_4} aims at illustrating the practical advantages of \textsc{bsp} in learning and forecasting several mortality  patterns characterized by a broad range of different evolutions across years and ages, ranging from smooth trajectories to more rapid shocks, over a broad time horizon. This motivates our focus on four illustrative countries, namely Italy, Sweden, United Kingdom and United States, whose gender--specific age--period log--mortality rates are available within the  \texttt{Human} \texttt{Mortality} \texttt{Database} (\url{https://www.mortality.org/}) for a wide time range which spans from $1933$ until either $2019$ or $2020$, and displays different country--specific evolutions of mortality over periods and age classes, along with fluctuations of varying magnitude due to shocks.

Although  the \texttt{Human} \texttt{Mortality} \texttt{Database} comprises data also for other countries, the corresponding time window is generally much shorter than  those studied in Section~\ref{sec_4}. In addition, Italy, Sweden, United Kingdom and United States exhibit a number of peculiar characteristics which make these countries of particular interest not only in forecasting, but also for inference. More specifically, as we will illustrate in Figures~\ref{fig:smoothing} and \ref{fig:covid}, Italy and the United Kingdom provide interesting examples to quantify the ability of the  proposed \textsc{bsp} in flexibly learning different magnitudes of the mortality shock, and the corresponding rates of change, associated with the World War {\rm II}. Recalling, e.g., \citet{vaupel1994longer}, Sweden is historically characterized by low mortality rates, but the recent evidence of slower rates of increment in the life expectancy relative to other countries \citep{drefahl2014losing} and the less stringent policy adopted during the \textsc{covid}--19 pandemic \citep[][]{wang2022estimating,juul2022mortality} make Sweden an interesting case study. The United States have  also experienced a slower increment in life expectancy in recent years, which culminated in a decreasing pattern over the past decade \citep{woolf2019life}. This specific behavior has motivated several explanatory studies mostly focused on peculiar mortality patterns and vulnerabilities associated with young and adult age classes \citep[e.g.,][]{remund2018cause,glei2021us}. These could also explain particular differences for the age--specific excess mortality in the United States during \textsc{covid}--19, relative to the patterns observed in other countries \citep[e.g.,][]{katzmarzyk2020obesity,wiemers2020disparities,goldstein2020demographic}. The  \textsc{bsp} formulation developed   in Sections~\ref{sec_2}--\ref{sec_3} is carefully designed to flexibly incorporate all these multifaceted  patterns and, therefore, the analysis of these four  countries provides a comprehensive setting to obtain empirical evidence of improved performance in inference and forecasting relative to state--of--the--art alternatives.

As highlighted in Section~\ref{sec_4}, the proposed   \textsc{bsp} yields improved forecasts also when applied to a different subgroup of countries from the \texttt{Human} \texttt{Mortality} \texttt{Database}, such as, for example, Czech Republic, Denmark and France. A discussion on future studies of  \textsc{bsp}  performance for low-- and middle--income countries, whose data are currently unavailable in the \texttt{Human} \texttt{Mortality} \texttt{Database}, can be found in  Section~\ref{sec_5}.

%%%%%%%%%%%%%%%%%%%%%%%%%%%%%%%%%%%%%%%%%%%%%%
%%%%%%%%%%%%%%%%%%%%%%%%%%%%%%%%%%%%%%%%%%%%%%
\section{Model Formulation}\label{sec_2}
Let $d_{xt}$ and $\textsc{e}_{xt}$ be the total death counts and the average number of individuals at risk (also known as {\em central exposed to risk})  at age $x$ in period $t$, respectively, within a given population. Following the overarching focus in the literature on mortality modeling \citep[e.g.,][]{booth2008mortality,hunt2021structure} our aim is to improve inference and forecasting of the observed central mortality rates defined as $m_{xt}=d_{xt}/\textsc{e}_{xt}$. 

To this end, let $\overline{m}_{xt}=\E(m_{xt} \mid \overline{m}_{xt})$ denote the underlying expected mortality rate at age~$x$ within period~$t$. We  introduce in Section~\ref{sec_21} a Poisson log--normal model for $d_{xt}$, whose rate parameter $\textsc{e}_{xt}\overline{m}_{xt}$ is allowed to flexibly vary across both ages $x$ and periods $t$ through a novel \textsc{b}--spline process for the expectation $f_t(x)$ of $\log \overline{m}_{xt}$. As clarified in Section~\ref{sec_22}, such a model admits a provably--accurate Gaussian state--space approximation which expresses the observed log--mortality rates $\log m_{xt}=\log (d_{xt}/\textsc{e}_{xt})$ via a linear combination of \textsc{b}--spline bases whose dynamic coefficients and the associated derivatives vary in time through a system of Gaussian state equations. This allows closed--form filtering, smoothing and forecasting of the coefficients trajectories and, as a consequence, of the induced patterns in the log--mortality rates $\log m_{xt}$ via a direct application of standard Kalman filter updates \citep{kalman1960new}; see Section~\ref{sec_3}.

%%%%%%%%%%%%%%%%%%%%%%%%%%%%%%%%%%%%%%%%%%%%%%
\subsection{B--Spline Process with Locally--Adaptive Dynamic Coefficients}\label{sec_21}
Recalling the above discussion, we model the death counts  $d_{xt}$, at each age $x \in \mathcal{X} \subset \mathbb{R}^{+}$ and period $t \in \mathcal{T} \subset \mathbb{R}^{+}$ via the Poisson log--normal distribution
\begin{equation}
(d_{xt} \mid \overline{m}_{xt}) \stackrel{\mbox{\scriptsize ind}}{\sim} \mbox{Poisson}(\textsc{e}_{xt}\overline{m}_{xt}), \quad \mbox{with} \quad (\log \overline{m}_{xt} \mid f_{t}(x)) \stackrel{\mbox{\scriptsize ind}}{\sim} \mbox{N}(f_{t}(x), \sigma_{\overline{m}}^2),
\label{eq1}
\end{equation}
for every $x\in \mathcal{X}$ and $t\in \mathcal{T}$, where $f_{t}(x)$ denotes a flexible function of age $x$ whose shape is allowed to vary with  $t$, whereas $ \sigma_{\overline{m}}^2$ encodes the global amount of over--dispersion in the observed death counts. The Poisson log--normal assumption in \eqref{eq1} has been considered in \citet{wong2018bayesian} to account for extra variability in the Poisson Lee--Carter model proposed by \citet{brouhns2002poisson} and \citet{czado2005bayesian}. Although this is a sensible modification which allows to formally incorporate age--specific heterogeneity in period mortality --- possibly arising from differences in cohort effects --- \citet{wong2018bayesian} still rely on the classical \citet{lee1992modeling}  parametric bilinear form for $f_{t}(x)$. As illustrated in Table~\ref{tab} (see column \textsc{lc}), such a form yields an overly--restrictive characterization of age--period mortality patterns that affects both inference and forecasting performance; see also \citet{delwarde2007smoothing} for an additional example of a Poisson  log--bilinear model that employs the classical \citet{lee1992modeling} construction.

To address the above issues and incorporate the core patterns of mortality discussed in Section~\ref{sec_1} and illustrated in Figure~\ref{fig:data}, we combine model \eqref{eq1} with a flexible, yet interpretable, representation for $f_{t}(x)$ based on a novel \textsc{b}--spline process with locally--adaptive dynamic coefficients. This formulation defines $f_{t}(x)$ through a linear combination of $p$ pre--selected \textsc{b}--spline basis functions of age, $g_1(x), \ldots, g_p(x)$,  whose associated coefficients $\beta_1(t), \ldots, \beta_p(t)$  jointly evolve over time via a system of stochastic differential equations that induce locally--varying smoothness and borrowing of information across contiguous bases. In particular, let $\bb(t)=[\beta_1(t),\partial \beta_1(t)/\partial t, a_1(t), \ldots, \beta_p(t),\partial \beta_p(t)/\partial t, a_p(t)]^{\intercal}$ be the $(3p \times 1)$--dimensional vector comprising the \textsc{b}--splines coefficients $\beta_1(t), \ldots, \beta_p(t)$, along with the associated first  derivatives $\partial\beta_1(t)/\partial t, \ldots,\partial \beta_p(t)/\partial t$, and the local instantaneous mean functions $a_1(t), \ldots, a_p(t)$ which induce time--varying smoothness  by controlling the expected value of the second derivatives  at time $t$, namely $a_j(t)=\E[\partial^2 \beta_j(t)/\partial^2 t \mid a_j(t)]$, for  $j=1, \ldots, p$. Moreover,  denote with $\bvarepsilon_t=[\varepsilon_{\beta_1}(t), \varepsilon_{a_1}(t), \ldots, \varepsilon_{\beta_p}(t), \varepsilon_{a_p}(t)]^{\intercal}$ a $(2p \times 1)$--dimensional vector encoding independent Gaussian white noise processes. Then, leveraging these quantities and letting $\tau=t/\lambda$ be a selected reference time scale, the proposed \textsc{bsp} assumes 
\begin{eqnarray}
f_t(x)&=& \sum\nolimits_{j=1}^p \beta_j(t)g_j(x), \qquad \  \ \quad \mbox{for any } x \in \mathcal{X} \ \mbox{and} \ t \in \mathcal{T}, \label{eq2} \\
\partial \bb(\lambda\tau)/\partial \tau&=& \lambda ({\bf I}_p \otimes \bC) \bb(\lambda\tau)+ ({\bf I}_p \otimes \bD)(\bOmega^{1/2}\bvarepsilon_{\tau}), \qquad \mbox{for any } \lambda \tau=t \in \mathcal{T},
\label{eq3}
\end{eqnarray}
where ${\bf I}_p $ is the $p \times p$ identity matrix,  $\otimes$ denotes the Kronecker product, $\lambda>0$ corresponds to a length--scale parameter that allows to preserve time unit invariance, $\bOmega$ is a suitably--specified  $2p \times 2p$ correlation matrix inducing local borrowing of information across contiguous splines coefficients --- via the correlation among the corresponding derivatives and local instantaneous means --- whereas $\bC$ and $\bD$  are system matrices defined as 
\begin{eqnarray}
\bC=\begin{bmatrix}
0 & 1 & 0 \\
0 & 0 & 1 \\
0 & 0 & 0
\end{bmatrix}, \qquad \qquad 
\bD=\begin{bmatrix}
0 & 0 \\
\sigma_{\beta} & 0 \\
0 & \sigma_a
\end{bmatrix},
\label{eq4}
\end{eqnarray}
with $\sigma_{\beta}>0$ and $\sigma_a>0$ denoting two  scale parameters. As clarified in \eqref{eq4},  these matrices are pre--specified to induce the desired system of stochastic differential equations; see   \citet[][A.6]{zhu2013locally} for a related definition of $\bC$ and $\bD$ in the univariate case.

\begin{figure}[t]
	\centering
		\includegraphics[width=0.88\textwidth]{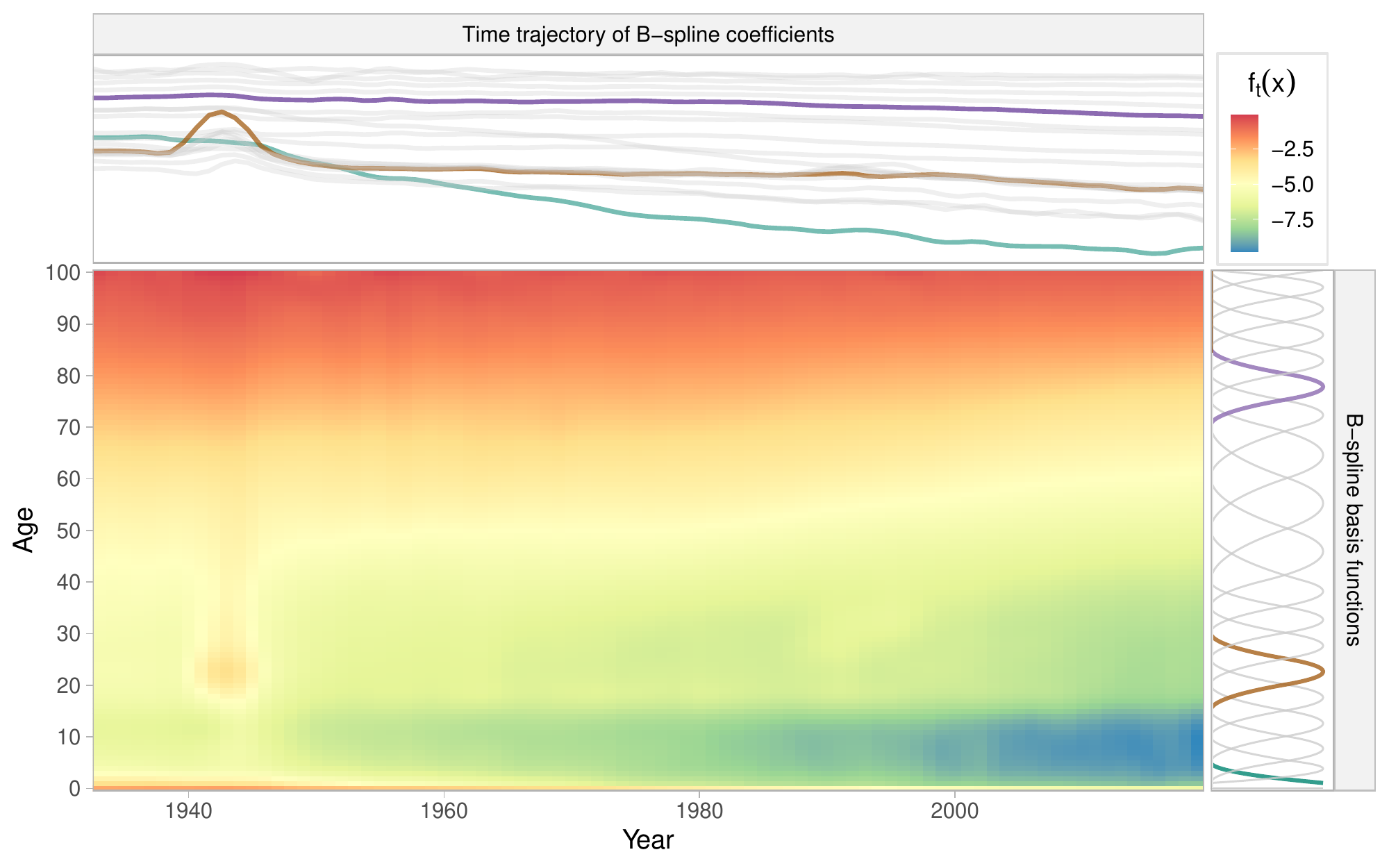}
		\caption{\footnotesize Illustrative example of a heatmap  for $f_t(x)$, as defined in \eqref{eq2} via a linear combination of pre--specified \textsc{b}--spline bases (right--side panel) with coefficients varying across periods according to \eqref{eq3}--\eqref{eq4}  (top--side panel). For illustrative purposes, three \textsc{b}--spline bases and the corresponding coefficients trajectories are highlighted with different colors; see Section~\ref{sec_4} for details on the choice of the number and location of the \textsc{b}--spline bases. Refer to the online article for the color version of this figure.}
	\label{fig:model}
\end{figure}

As illustrated in Figure~\ref{fig:model}, the \textsc{b}--spline process representation in equations \eqref{eq2}--\eqref{eq4} provides an effective formulation which treats $f_t(x)$ as a function of age $x$, through a linear combination of interpretable \textsc{b}--spline bases directly associated to specific age classes, and as a stochastic process of time $t=\lambda \tau$, leveraging a flexible system of stochastic differential equations that jointly characterize the time trajectory of every spline coefficient $\beta_j(t)$, for $j=1, \ldots, p$, by explicitly modeling its smoothness across periods.  In \eqref{eq3}--\eqref{eq4}, such a smoothness is measured by the second--order derivative $\partial^2 \beta_j(t)/\partial^2 t$ which is in turn centered on a higher level time--varying instantaneous mean function $a_j(t)$ that allows local adaptivity. Combining  \eqref{eq2}--\eqref{eq4} with  model \eqref{eq1} yields a unique representation for age--period mortality patterns that (i) accounts for age–specific heterogeneity in death counts via the log--normal assumption in  \eqref{eq1}, (ii) enforces a generally smooth trajectory for the age patterns of mortality through the linear combination of   \textsc{b}--splines in \eqref{eq2}, and (iii) explicitly allows these patterns to evolve in time between periods of rapid and slow variations, affecting age classes with different magnitudes, via the system of stochastic differential equations in  \eqref{eq3}--\eqref{eq4} for the splines coefficients. 

To further clarify representation \eqref{eq3}--\eqref{eq4}, it shall be emphasized that this construction extends the nested Gaussian process of \citet{zhu2013locally} in a number of directions inherently motivated by our focus on modeling and forecasting of mortality rates. In fact, the original formulation by \citet{zhu2013locally} does not consider  \eqref{eq2}, and provides a simpler version of  \eqref{eq3}--\eqref{eq4}  with a focus on inducing locally--varying smoothness in a single trajectory via a Gaussian process   \citep[see e.g.,][]{williams2006gaussian} for a derivative of selected order, which is in turn centered on a higher--level Gaussian process characterizing the local instantaneous mean function. Although considering separate nested Gaussian processes for the trajectories $f_t(x)$ of each age $x \in \mathcal{X}$ is a viable strategy, such a representation  is affected by the choice of the time scale and, more crucially, it fails to borrow information among mortality patterns for contiguous ages. Recalling, e.g., \citet{currie2004smoothing} and \citet{camarda2019smooth} the latter property would be conceptually and practically useful since it is reasonable to expect that the mortality patterns of close age classes display a natural dependence. To this end, equations \eqref{eq2}--\eqref{eq4} extend \citet{zhu2013locally} to a structured multivariate formulation which induces such a borrowing of information through the  \textsc{b}--spline representation of $f_t(x)$ in \eqref{eq2} and by inducing dependence among the \textsc{b}--spline coefficients in \eqref{eq3} via the introduction of correlation between the white noises through the matrix $\bOmega$. This symmetric matrix has unit diagonal, and off--diagonal elements that are non--zero only for the entries $\bOmega_{j,l}$ whose indexes $(j,l)$ are either both even or both odd, so as to induce correlation among the noises associated with the  derivatives and local instantaneous means, respectively. As clarified in Section~\ref{sec_3}, by defining these non--zero correlations via suitable covariance functions \citep[e.g.,][]{williams2006gaussian} allows to enforce a local borrowing of information which decays as the distance between age classes grows. The introduction of the length--scale parameter $\lambda$ allows, instead, to preserve time unit invariance, so that, if the reference time scale is changed it is still possible to retrieve the same model formulation with a suitable specification of the parameters $\lambda$, $\sigma_{\beta}$ and $\sigma_a$. As clarified in  Proposition~\ref{p2} and in the subsequent discussion, this property follows from the fact that the process in \eqref{eq3}--\eqref{eq4} has been defined with respect to a reference time scale $\tau$ related with $t$ via $\tau=t/\lambda$. This modification is in line with similar operations considered in the Gaussian process literature when including a length--scale parameter in popular covariance functions  \citep[][]{williams2006gaussian}. The empirical results in Section~\ref{sec_4} confirm that these extensions yield substantial gains in mortality forecasts relative to those obtained via a direct application of the original nested Gaussian process by \citet{zhu2013locally}  to each trajectory $f_t(x)$, $t \in \mathcal{T}$, separately for every $x$.

Besides including the core age--period structures of mortality, model  \eqref{eq1}--\eqref{eq4} crucially admits a provably accurate Gaussian state--space approximation, as described in Section~\ref{sec_22} below. This representation further clarifies the proposed model and facilitates efficient computation via standard Kalman filter updates; see also Section~\ref{sec_3}, and refer to Section~\ref{sec_4} for details on the choice of the number and location of the \textsc{b}--spline bases.

%%%%%%%%%%%%%%%%%%%%%%%%%%%%%%%%%%%%%%%%%%%%%%

\subsection{Gaussian State--Space Approximation}\label{sec_22}
As a first step towards the derivation of an accurate and computationally tractable Gaussian state--space representation of the proposed formulation in \eqref{eq1}--\eqref{eq4}, Proposition~\ref{p1} proves that model  \eqref{eq1} induces a distribution on the observed log--mortality rates  $\log m_{xt}=\log (d_{xt}/\textsc{e}_{xt})$ which can be closely approximated, for $\textsc{e}_{xt}$ large enough, by the $\mbox{N}(f_{t}(x), \sigma_{\overline{m}}^2)$ assumed in  \eqref{eq1} for $\log \overline{m}_{xt}$. See Appendix A for proofs.

\begin{Proposition}
Under model \eqref{eq1}, $\log m_{xt}=\log (d_{xt}/\textsc{e}_{xt}) \to \mbox{\normalfont N}(f_{t}(x), \sigma_{\overline{m}}^2)$ in distribution, as $\textsc{e}_{xt} \to~\infty$, for any $x \in \mathcal{X}$ and $t\in \mathcal{T}$.
\label{p1}
\end{Proposition}

Proposition~\ref{p1} motivates direct focus on the observed log--mortality rates $\log m_{xt}$, which are of overarching interest in state--of--the--art studies \citep[see, e.g.,][]{land1986methods,booth2008mortality,currie2016fitting,hunt2021structure}. In addition, it justifies the adoption of the Gaussian regression model $\log m_{xt} = f_{t}(x) + \nu_{xt}$, with \smash{$\nu_{xt} \stackrel{\mbox{\scriptsize i.i.d.}}{\sim} \mbox{N}(0, \sigma_{\overline{m}}^2)$}  for $x \in \mathcal{X}$, $t\in \mathcal{T}$, and  $f_{t}(x)$ as in  \eqref{eq2}, which is  more tractable than the Poisson log--normal representation for the death counts in \eqref{eq1}. Recalling Proposition~\ref{p1}, this approximation is provably accurate  in settings with  large enough $\textsc{e}_{xt}$. This is a common situation in mortality studies by country, where $\textsc{e}_{xt}$ is typically in the order of tens--to--hundreds of thousands.

Although Proposition~\ref{p1} yields a simpler construction, to obtain a fully tractable formulation it is  necessary to derive an alternative representation for the stochastic differential equations in \eqref{eq3}--\eqref{eq4} which is amenable to efficient computation, direct forecasting, interpretable inference and principled uncertainty quantification. Proposition~\ref{p2} proves that, when observed at a finite collection of times $t_1, \ldots,t_n$, as in our mortality--data context, equations \eqref{eq3}--\eqref{eq4} admit a tractable representation via a linear system of Gaussian state equations.

\begin{Proposition}
Let $\bb_{t_s}$ denote the realization at a generic time $t_s$ of the process~$\bb(t)$ defined in Section~\ref{sec_21}, i.e.,  $\bb_{t_s}=[\beta_1(t),\partial \beta_1(t)/\partial t, a_1(t), \ldots, \beta_p(t),\partial \beta_p(t)/\partial t, a_p(t)]_{| t=t_s}^{\intercal}$. Then, for each finite grid of times $t_s=t_1, \ldots,t_n$, with $t_1 < \cdots < t_n$, the system of stochastic differential equations in \eqref{eq3}--\eqref{eq4} admits the Gaussian state--equation representation
\begin{eqnarray}
\bb_{t_{s+1}}&=& \bT_{t_s} \bb_{t_s}+\boeta_{t_s}, \quad \boeta_{t_s} \stackrel{\mbox{\scriptsize \normalfont ind}}{\sim} \mbox{\normalfont N}_{3p}({\bf 0}, \bQ_{t_s}), \qquad \mbox{\normalfont for } \ t_s=t_1, \ldots,t_n,
\label{eq5}
\end{eqnarray}
where $\bT_{t_s}$ denotes a $3p \times 3p$ block--diagonal transition matrix defined as
\begin{eqnarray}
\scalemath{1}{\bT_{t_s}=\setlength\arraycolsep{5pt} {\bf I}_p \otimes\begin{bmatrix}
1 & \lambda\delta_s & \lambda^2(\delta_s^2/2) \\
0 & 1 & \lambda\delta_s \\
0 & 0 & 1
\end{bmatrix},}
\label{eq6}
\end{eqnarray}
with $\delta_s=(t_{s{+}1}- t_{s})/\lambda$, while $ \bQ_{t_s}$ is a $3p \times 3p$ covariance matrix having generic  block 
\begin{equation}
\scalemath{1}{\bQ_{t_s[j,l]}=\setlength\arraycolsep{5pt} \sigma^2_\beta \rho_{\beta[j,l]}  \begin{bmatrix}
(\delta^3_s/3)\lambda^2  & 
(\delta^2_s/2)\lambda & 0 \\
(\delta^2_s/2)\lambda &
\delta_s&0\\
0 & 0  & 0
\end{bmatrix}+\sigma^2_a \rho_{a[j,l]}\begin{bmatrix}
(\delta^5_s/20)\lambda^4  & 
(\delta^4_s/8)\lambda^3 & (\delta^3_s/6)\lambda^2 \\
(\delta^4_s/8)\lambda^3 &
(\delta^3_s/3)\lambda^2&(\delta^2_s/2)\lambda\\
(\delta^3_s/6)\lambda^2& (\delta^2_s/2)\lambda  & \delta_s
\end{bmatrix},}
\label{eq7}
\end{equation}
for each $j=1, \ldots, p$ and $l=1, \ldots, p$, with $\rho_{\beta[j,l]}$ and $\rho_{a[j,l]}$ denoting  those entries of $\bOmega$ in~\eqref{eq3} that measure the correlation among the derivatives of the coefficients associated with splines $j$ and $l$, and between the corresponding local instantaneous means, respectively.
\label{p2}
\end{Proposition}

Proposition~\ref{p2} yields a simple state--space representation expressing the value of $\beta_j(t_{s+1})$ at time $t_{s+1}$ via a second--order stochastic Taylor expansion of the trajectory $\beta_j(t)$ around the previous time point $t_{s}$, for each $j=1, \ldots, p$; see the form of  $\bT_{t_s}$ in \eqref{eq6}. This allows to explicitly model and forecast not only the class--specific mortality trends encoded in $\beta_1(t), \ldots, \beta_p(t)$, but also the associated rates of change measured by $\partial \beta_1(t)/\partial t, \ldots, \partial \beta_p(t)/\partial t$ and the corresponding instantaneous means $a_1(t), \ldots, a_p(t)$. Recalling the discussion in Section~\ref{sec_21}, notice that the specification of the actual values of $t_s=t_1, \ldots,t_n$ is not necessary in    equations \eqref{eq5}--\eqref{eq7}, which only require to pre--specify the time lags $\delta_s$ among consecutive observations under the reference scale $t/\lambda$. Changing such a scale, and hence the lags $\delta_s$, does not alter the model, provided that it is possibile to modify $\lambda$, $\sigma^2_\beta$ and $\sigma^2_a$ accordingly in order to preserve the same state equations in $t$. This is not possibile under the state equations displayed in Section 3 of \citet{zhu2013locally}. In our equally--spaced context we  set lags $\delta_s$ to $1$, and then learn the appropriate scaling  $\lambda$ and variances via maximum likelihood under \eqref{eq5}--\eqref{eq7}.

The state equations in \eqref{eq5}--\eqref{eq7} are both flexible and interpretable, and further allow to borrow information across the coefficients of different \textsc{b}--splines via the covariance matrix $\bQ_{t_s}$ of the noise $ \boeta_{t_s}$;  see \eqref{eq7}. The core parameters regulating the strength of such a dependence are $\rho_{\beta[j,l]}$ and $\rho_{a[j,l]}$,  for $j=1, \ldots, p$ and $l=1, \ldots, p$. Letting $\rho_{\beta[j,l]}=\mathbbm{1}(j=l)$ and $\rho_{a[j,l]}=\mathbbm{1}(j=l)$ yields no borrowing of information and, as a consequence, separate state equations for each \textsc{b}--spline coefficient. Conversely, whenever $\rho_{\beta[j,l]} \in (0{,}1]$ and $\rho_{a[j,l]} \in (0{,}1]$, the $j$--th and $l$--th splines display a dependence in the trajectories of the associated coefficients. More specifically, large values of $\rho_{\beta[j,l]}$ and $\rho_{a[j,l]}$ imply high correlation between the first derivatives and local instantaneous means functions, respectively, of the coefficient trajectories for splines $j$ and $l$. This allows to borrow information in terms of  both  the overall trend and smoothness, while inducing dependence among the actual trajectories $\beta_j(t)$ and $\beta_l(t)$ under \eqref{eq5}--\eqref{eq7}.

Recalling the above discussion and extending related ideas from \textsc{p}--splines representations \citep{eliers1996flexible,lang2004bayesian}, we define $\rho_{\beta[j,l]}$ and $\rho_{a[j,l]}$ to induce a local borrowing  of information whose strength decays with a suitable distance between the  $j$--th and $l$--th splines. More specifically, let $\bar{x}_j$ and $\bar{x}_l$ denote the ages at which the \textsc{b}--spline functions $g_j(x)$ and $g_l(x)$ are maximized, respectively, we define
\begin{eqnarray}
\rho_{\beta[j,l]}=\mathcal{K}(\bar{x}_j,\bar{x}_l; \bgamma_\beta), \qquad \mbox{and} \qquad \rho_{a[j,l]}=\mathcal{K}(\bar{x}_j,\bar{x}_l; \bgamma_a),
\label{eq8}
\end{eqnarray}
for every spline $j=1, \ldots, p$ and $l=1, \ldots, p$, where $\mathcal{K}(\bar{x}_j,\bar{x}_l; \bgamma_{\beta})$ and $\mathcal{K}(\bar{x}_j,\bar{x}_l; \bgamma_a)$ denote user--selected covariance functions  \citep[e.g.,][Ch.\ 4]{williams2006gaussian}, which decay to zero as $|\bar{x}_j-\bar{x}_l|$ grows, and are defined such that $\mathcal{K}(\bar{x}_j,\bar{x}_j; \bgamma_\beta)=\mathcal{K}(\bar{x}_j,\bar{x}_j; \bgamma_a)=1$ for $j=1, \ldots, p$. As a consequence, the time patterns of mortality are allowed to effectively share local information across contiguous age classes, and the strength of this dependence progressively decreases for  distant ages with a pattern which depends on the selected covariance functions and on the associated parameters $\bgamma_\beta$ and $\bgamma_a$. Routinely--implemented examples of covariance functions are the squared exponential and the Mat\'ern, among others \citep[see e.g.,][Ch.\ 4]{williams2006gaussian}; see Section~\ref{sec_4} for details on suitable specifications of these covariance functions and the  corresponding parameters in the mortality data context.

Combining Propositions~\ref{p1} and \ref{p2} yields the tractable and accurate Gaussian state--space approximation for the observed log--mortality rates 
\begin{eqnarray}
\log \bm_{t_s}&=&\bZ_{t_s}\bb_{t_s} + \bnu_{t_s}, \qquad \bnu_{t_s} \stackrel{\mbox{\scriptsize \normalfont ind}}{\sim} \mbox{\normalfont N}_{k}({\bf 0}, \bH_{t_s}), \label{eq9}\\
\bb_{t_{s+1}}&=& \bT_{t_s} \bb_{t_s}+\boeta_{t_s}, \qquad \boeta_{t_s} \stackrel{\mbox{\scriptsize \normalfont ind}}{\sim} \mbox{\normalfont N}_{3p}({\bf 0}, \bQ_{t_s}), 
\label{eq10}
\end{eqnarray}
for every time $t_s=t_1, \ldots,t_n$, where $\log \bm_{t_s}=(\log m_{x_1,t_s}, \ldots, \log m_{x_k,t_s})^{\intercal}$   is the $(k \times 1)$--dimensional vector of the log--mortality rates observed  for ages $x_1, \ldots, x_k$ at time $t_s$, $\bZ_{t_s}=[\bg_1{,}{\bf 0}, {\bf 0}, \bg_2{,}{\bf 0}, {\bf 0}, \ldots, \bg_p{,}{\bf 0}, {\bf 0}]$ denotes the $(k \times 3p)$--dimensional design matrix with non--zero columns $\bg_j=[g_j(x_1), \ldots, g_j(x_k)]^{\intercal}$, $j=1, \ldots, p$ comprising the values of the pre--selected \textsc{b}--splines bases at the observed ages $x_1, \ldots, x_k$, and $\bH_{t_s}=\sigma_{\overline{m}}^2{\bf I}_k$, whereas $\bb_{t_s}$, $\bT_{t_s}$ and $\bQ_{t_s}$ are defined as in Proposition~\ref{p2}.

As clarified in Section~\ref{sec_3}, the above Gaussian state--space approximation allows closed--form filtering, smoothing and forecasting via simple recursive equations obtained from direct application of the classical Kalman filter updates  \citep{kalman1960new,koopman2000fast}; see also \citet{durbin2012time} and \citet{chopin2020introduction} for a general treatment of the Kalman filter and smoother in linear Gaussian state--space models, and refer to the \texttt{R} package \texttt{KFAS} \citep{helske2017kfas} for an effective implementation. This  tractability is in contrast with recently--proposed flexible mortality models which require \textsc{mcmc} routines \citep[e.g.,][]{wong2018bayesian,alexopoulos2019bayesian} and, unlike for state--of--the--art formulations discussed in Section~\ref{sec_1}, model \eqref{eq9}--\eqref{eq10} holds not only for equally--spaced time grids $t_1, \ldots,t_n$ but also for unequally--spaced ones. Such a generality is conceptually and practically useful in allowing inference and forecasting for different time horizons, which is of interest, for example, during periods of mortality shocks to rapidly revise forecasts in the short term, e.g., within months, trimesters or semesters.

%%%%%%%%%%%%%%%%%%%%%%%%%%%%%%%%%%%%%%%%%%%%%%
%%%%%%%%%%%%%%%%%%%%%%%%%%%%%%%%%%%%%%%%%%%%%%
\section{Filtering, Smoothing and Forecasting}\label{sec_3}
In Section~\ref{sec_31} we leverage the model in \eqref{eq9}--\eqref{eq10} to derive tractable strategies for probabilistic inference and prediction of the coefficients vector $\bb_{t_s}$ defined in Proposition~\ref{p2} and, as a direct consequence, of  the $(k \times 1)$--dimensional log--mortality rates mean vector $\bef_{t_s}$ defined as $\bef_{t_s}=\bZ_{t_s}\bb_{t_s}=[f_{t_s}(x_1), \ldots, f_{t_s}(x_k)]^{\intercal}$, with each $f_{t}(x)$ as in equation~\eqref{eq2}. To this end, we employ the classical Kalman filter \citep{kalman1960new} under the model in \eqref{eq9}--\eqref{eq10}  to obtain simple and closed--form recursive formulas for the filtering $p(\bb_{t_s} \mid \log \bm_{t_1}, \ldots, \log \bm_{t_s})$, predictive $p(\bb_{t_{s+1}} \mid \log \bm_{t_1}, \ldots, \log \bm_{t_s})$, and smoothing $p(\bb_{t_s} \mid \log \bm_{t_1}, \ldots, \log \bm_{t_n})$ distributions of $\bb_{t_s}$, for $t_s=t_1, \ldots,t_n$. Since $\bef_{t_s}=\bZ_{t_s}\bb_{t_s}$, with $\bZ_{t_s}$ known, the filtering, predictive and smoothing distributions for $\bef_{t_s}$ can be directly derived from those of $\bb_{t_s}$, for every $t_s=t_1, \ldots,t_n$. Leveraging these results, we further develop in Section~\ref{sec_32} a modern version of the celebrated  \citet{lee1992modeling} approach. Our proposed strategy provides future probabilistic projections of the \textsc{b}--splines coefficients via a simple random walk plus drift model where the drift component  exploits the possibility of our formulation to explicitly learn not only mortality levels but also the corresponding rates of change. As illustrated in Section~\ref{sec_4}, this solution yields improved probabilistic forecasts of log--mortality rates relative to state--of--the--art alternatives.

The above formulas require  knowledge of the parameters $\sigma_{\overline{m}}^2$, $\sigma_{\beta}^2$, $\sigma_{a}^2$ and $\lambda$. Thanks to the Gaussian form of the model in \eqref{eq9}--\eqref{eq10}, these quantities can be estimated via maximization of the marginal likelihood for the Gaussian vectors $(\log \bm_{t_1}, \ldots, \log \bm_{t_n})$, which is available in closed form, thereby allowing direct estimation; see \citet[][Ch. 7]{durbin2012time} and  \citet[][Ch. 7]{chopin2020introduction} for further details on maximum likelihood estimation of the system parameters in Gaussian state--space models, and refer to the \texttt{R} package \texttt{KFAS} \citep{helske2017kfas} for an effective implementation. As discussed in the tutorial code at \url{https://github.com/fpavone/BSP-mortality}, in practice it is often recommended to add suitable penalizations and initialize  the estimation procedure at different starting points, selecting as final estimate the one that yields the highest marginal likelihood, thereby avoiding possible issues associated with local modes. The covariance function parameters $\bgamma_\beta$ and $\bgamma_a$ in \eqref{eq8} are instead fixed at default values that allow to induce a suitable borrowing of information across spline coefficients; see Section~\ref{sec_4} for details.

%%%%%%%%%%%%%%%%%%%%%%%%%%%%%%%%%%%%%%%%%%%%%%
\subsection{Filtering, Predictive and Smoothing Distributions}\label{sec_31}
Due to the Gaussian form of  model  \eqref{eq9}--\eqref{eq10}, the filtering, predictive and smoothing distributions for  $\bb_{t_s}$ are multivariate normals $\mbox{N}_{3p}(\bmu_{t_s|t_{1:s}}, \bSigma_{t_s|t_{1:s}})$, $\mbox{N}_{3p}(\bmu_{t_{s+1}|t_{1:s}}, \bSigma_{t_{s+1}|t_{1:s}})$ and $\mbox{N}_{3p}(\bmu_{t_s|t_{1:n}}, \bSigma_{t_s|t_{1:n}})$, respectively, with the mean vectors and covariance matrices that can be derived sequentially in time via recursive equations \citep{kalman1960new}. More specifically, let $\bmu_{t_{s}|t_{1:s-1}}$ and $ \bSigma_{t_{s}|t_{1:s-1}}$ be the predictive mean vector and covariance matrix for  $\bb_{t_s}$ given the log--mortality rates  observed until time $t_{s-1}$. Then, recalling, e.g.,  \citet[][Ch.\ 4]{durbin2012time},  the filtering distribution for $\bb_{t_s}$ is a $3p$--variate Gaussian with mean vector $\bmu_{t_{s}|t_{1:s}}$ and covariance matrix $\bSigma_{t_{s}|t_{1:s}}$ equal to
\begin{eqnarray}
\begin{split}
\qquad \bmu_{t_{s}|t_{1:s}}&=\bmu_{t_{s}|t_{1:s-1}}+ \bSigma_{t_{s}|t_{1:s-1}}\bZ_{t_s}^{\intercal}(\bZ_{t_s} \bSigma_{t_{s}|t_{1:s-1}}\bZ_{t_s}^{\intercal}+\bH_{t_s})^{-1}(\log \bm_{t_s}-\bZ_{t_s}\bmu_{t_{s}|t_{1:s-1}}),\\
\bSigma_{t_{s}|t_{1:s}}&= \bSigma_{t_{s}|t_{1:s-1}}-\bSigma_{t_{s}|t_{1:s-1}}\bZ_{t_s}^{\intercal}(\bZ_{t_s} \bSigma_{t_{s}|t_{1:s-1}}\bZ_{t_s}^{\intercal}+\bH_{t_s})^{-1}\bZ_{t_s}\bSigma_{t_{s}|t_{1:s-1}}.
\end{split}
\label{eq11}
\end{eqnarray}
The above results are a direct consequence of the closure under conditioning of multivariate Gaussians and, when combined with \eqref{eq10}, directly yield the mean vector $\bmu_{t_{s+1}|t_{1:s}}$ and the covariance matrix $\bSigma_{t_{s+1}|t_{1:s}}$ for the predictive Gaussian distribution of $\bb_{t_{s+1}}$. More specifically, leveraging the closure under linear combinations  of multivariate Gaussians, we obtain
\begin{eqnarray}
\begin{split}
\bmu_{t_{s+1}|t_{1:s}}&=\bT_{t_{s}}\bmu_{t_{s}|t_{1:s}},\\
\bSigma_{t_{s+1}|t_{1:s}}&= \bT_{t_{s}}\bSigma_{t_{s}|t_{1:s}}\bT^{\intercal}_{t_{s}}+\bQ_{t_s}.
\end{split}
\label{eq12}
\end{eqnarray}
Equations \eqref{eq11}--\eqref{eq12} provide simple closed--form formulas that allow to filter and forecast $\bb_{t_s}$ recursively from time $t_1$ until time $t_n$ by iterating among filtering and one--step--ahead predictive steps. Recalling, e.g.,  \citet[][Ch.\ 4]{durbin2012time}, such a recursion is initialized at $t_1$ from a $\mbox{N}_{3p}(\bmu_{t_{1}|t_{0}}, \bSigma_{t_{1}|t_{0}})$. Although several starting strategies can be considered  \citep[see e.g.,][]{durbin2012time}, we rely on a data--driven approach and  fix $\bmu_{t_{1}|t_{0}}$ at a frequentist estimate based on a simple spline regression, while  $\bSigma_{t_{1}|t_{0}}$ is set to $10{\bf I}_{3p}$ to induce a relatively diffuse initialization.
 
The forward recursions in  \eqref{eq11}--\eqref{eq12} can be also combined with backward iterations to derive the mean vector $\bmu_{t_{s}|t_{1:n}}$ and covariance matrix $\bSigma_{t_{s}|t_{1:n}}$ of the Gaussian smoothing distribution for each $t_s$. This is obtained by iterating backward in time from $t_n$ to $t_1$ via the expressions
\begin{eqnarray}
\begin{split}
\bmu_{t_{s}|t_{1:n}}&=\bmu_{t_{s}|t_{1:s-1}}+\bSigma_{t_{s}|t_{1:s-1}}\br_{t_{s-1}},\\
\bSigma_{t_{s}|t_{1:n}}&= \bSigma_{t_{s}|t_{1:s-1}}-\bSigma_{t_{s}|t_{1:s-1}}\bV_{t_{s-1}}\bSigma_{t_{s}|t_{1:s-1}},
\end{split}
\label{eq13}
\end{eqnarray}
where $\br_{t_{s-1}}$ and $\bV_{t_{s-1}}$ are obtained from the backward equations $\br_{t_{s-1}}=\bZ^{\intercal}_{t_s}(\bZ_{t_s} \bSigma_{t_{s}|t_{1:s-1}}\bZ_{t_s}^{\intercal}+\bH_{t_s})^{-1}(\log \bm_{t_s}-\bZ_{t_s}\bmu_{t_{s}|t_{1:s-1}})+\bL^{\intercal}_{t_s}\br_{t_{s}}$ and $\bV_{t_{s-1}}=\bZ^{\intercal}_{t_s}(\bZ_{t_s} \bSigma_{t_{s}|t_{1:s-1}}\bZ_{t_s}^{\intercal}+\bH_{t_s})^{-1}\bZ_{t_s}+\bL^{\intercal}_{t_s}\bV_{t_{s}}\bL_{t_s}$, with $\bL_{t_s}=\bT_{t_s}-\bT_{t_s} \bSigma_{t_{s}|t_{1:s-1}} \bZ^{\intercal}_{t_s}(\bZ_{t_s} \bSigma_{t_{s}|t_{1:s-1}}\bZ_{t_s}^{\intercal}+\bH_{t_s})^{-1} \bZ_{t_s}$, and initialization $\br_{t_{n}}={\bf 0}$ and $\bV_{t_{n}}={\bf 0}_{3p \times 3p}$; see \citet[][Ch.\ 4.4]{durbin2012time} for  a detailed derivation of \eqref{eq13} leveraging again standard properties of multivariate Gaussian distributions.

Since $\bef_{t_s}=\bZ_{t_s}\bb_{t_s}$, it immediately follows that the filtering, predictive and smoothing distributions for the log--mortality rates mean function $f_{t_s}(x)$, at $x=x_1, \ldots, x_k$, can be directly obtained from those of $\bb_{t_s}$ in equations~\eqref{eq11}, \eqref{eq12} and \eqref{eq13}, respectively. This implies
\begin{eqnarray}
\begin{split}
(\bef_{t_s} \mid \log \bm_{t_1}, \ldots, \log \bm_{t_s}) &\sim \mbox{N}_{k}(\bZ_{t_s}\bmu_{t_{s}|t_{1:s}}, \bZ_{t_s}\bSigma_{t_{s}|t_{1:s}}\bZ^{\intercal}_{t_s}),\\
(\bef_{t_{s+1}} \mid \log \bm_{t_1}, \ldots, \log \bm_{t_s}) &\sim \mbox{N}_{k}(\bZ_{t_{s+1}}\bmu_{t_{s+1}|t_{1:s}}, \bZ_{t_{s+1}}\bSigma_{t_{s+1}|t_{1:s}}\bZ^{\intercal}_{t_{s+1}}),\\
(\bef_{t_s} \mid \log \bm_{t_1}, \ldots, \log \bm_{t_n}) &\sim \mbox{N}_{k}(\bZ_{t_s}\bmu_{t_{s}|t_{1:n}}, \bZ_{t_s}\bSigma_{t_{s}|t_{1:n}}\bZ^{\intercal}_{t_s}),
\end{split}
\label{eq14}
\end{eqnarray}
for each $t_s=t_1, \ldots, t_n$. These results yield closed--form Gaussian distributions that facilitate probabilistic inference on mortality levels and the corresponding rates of change across ages and periods, beyond currently--available analyses.  Moreover, as clarified in Section~\ref{sec_32}, this perspective allows to improve point forecasts and predictive intervals of  future log--mortality rates at different  times. Crucially, these quantities  can be readily computed via user--friendly and optimized \texttt{R} packages for state--space models. For example, the filtering, predictive and smoothing distributions in \eqref{eq11}--\eqref{eq14} can be obtained via the  \texttt{KFS} function from the  \texttt{KFAS} package \citep{helske2017kfas}, after specifying  model \eqref{eq9}--\eqref{eq10} via the function \texttt{SSModel}.

%%%%%%%%%%%%%%%%%%%%%%%%%%%%%%%%%%%%%%%%%%%%%%
\subsection{Forecasting}\label{sec_32}
As is clear from equation  \eqref{eq9}, the results in Section~\ref{sec_31} are useful not only for inference, but also to obtain probabilistic forecasts for the vector of future log--mortality rates  $\log \bm_{t_{s^*}}=(\log m_{x_1,t_{s^*}}, \ldots, \log m_{x_k,t_{s^*}})^{\intercal}$ for $t_{s^*}=t_{n+1}, \ldots, t_{n+n^*}$, from the predictive distribution of $\bef_{t_s}$. While such an approach is expected to yield accurate results for short--term forecasts, the quantitative studies in Section~\ref{sec_4} suggest that the inherent local adaptivity of the model developed in Section~\ref{sec_2} might yield less stable and shock--robust mortality projections for those medium--to--large time horizons that are of interest in demography. 

To address this aspect and deliver accurate probabilistic forecasts at longer time horizons, we derive a simple, yet effective, strategy that combines the proposed \textsc{b}--spline construction in Section~\ref{sec_2} with the  celebrated random walk plus drift projections by  \citet{lee1992modeling}, in order to obtain an improvement in log--mortality rates forecasts relative to state--of--the--art competitors. Despite its simplicity, the random walk plus drift construction is empirically supported by the globally--linear trend of log--mortality rates in medium--to--large time horizons, which make  \citet{lee1992modeling} projections still competitive. Nonetheless, as mentioned in e.g.,  \citet{hyndman2007robust}, these  forecasts still rely on an age–period bilinear formulation that fails to account for heterogeneity in age--specific mortality dynamics, and, in addition, the estimation of the drift component is not robust to mortality shocks. The model we propose in Section~\ref{sec_2} accounts for both effects, thus suggesting that incorporating the random walk plus drift forecasting strategy within the proposed  \textsc{b}--spline process with locally--adaptive dynamic coefficients would yield improvements over the original \citet{lee1992modeling} strategy and, as clarified in Section~\ref{sec_4}, also with respect to other state--of--the--art  methods, both in terms of point forecasts and calibration of the predictive intervals.

Consistent with the above discussion, and recalling equations \eqref{eq2} and \eqref{eq9}, we obtain point forecasts for  the future log--mortality rates $\log \bm_{t_{s^*}}$, via
 \begin{eqnarray}
\hat{\bef}_{t_{s^*}}=[\bg_1, \ldots, \bg_p] \hat{\bbeta}_{t_{s^*}} \quad \mbox{for} \ \ t_{s^*}=t_{n+1}, \ldots, t_{n+n^*},
\label{eq15}
\end{eqnarray}
where  $[\bg_1, \ldots, \bg_p] $ corresponds to the $(k \times p)$--dimensional \textsc{b}--splines matrix having columns \smash{$\bg_j=[g_j(x_1), \ldots, g_j(x_k)]^{\intercal}$}, $j=1, \ldots, p$, whereas \smash{$\hat{\bbeta}_{t_{s^*}}=[\hat{\beta}_1(t_{s^*}), \ldots, \hat{\beta}_p(t_{s^*})]^{\intercal}$} is the $(p \times 1)$--dimensional vector comprising the forecasts for the \textsc{b}--splines coefficients at $t_{s^*}>t_n$ from the $p$--variate random walk plus drift model
 \begin{eqnarray}
 \begin{split}
\bbeta_{t_{s^*+1}} &= \bbeta_{t_{s^*}} +\hat{\lambda} \delta_{s^*}\bDelta_{t_{s^*}}+\bomega_{t_{s^*}},  \qquad \bomega_{t_{s^*}} \stackrel{\mbox{\scriptsize \normalfont i.i.d}}{\sim} \mbox{\normalfont N}_p( {\bf 0}, \bW), \\
\bDelta_{t_{s^*+1}}&=\bDelta_{t_{s^*}}+\bepsilon_{t_{s^*}}  \qquad \qquad \qquad \quad \ \  \bepsilon_{t_{s^*}}  \ \smash{\stackrel{\mbox{\scriptsize \normalfont i.i.d}}{\sim}} \ \mbox{\normalfont N}_p( {\bf 0}, \sigma_{\Delta}^2 {\bf I}_p).
\label{eq16}
\end{split}
\end{eqnarray}
In \eqref{eq16}, each $\bW_{[j,l]}$ is set equal to \smash{$\sigma_{\omega}^2\rho_{\beta[j,l]}$}, for  $j=1, \ldots, p$ and  $l=1, \ldots, p$, with $\rho_{\beta[j,l]}$ as in \eqref{eq8}, \smash{$\hat{\lambda}$} is the maximum marginal likelihood estimate of $\lambda$ discussed in Section~\ref{sec_3}, whereas the starting values \smash{$\hat{\bbeta}$} and  \smash{$\hat{\bDelta}$} for ${\bbeta}_{t_n}$ and  \smash{${\bDelta}_{t_n}=[{\Delta}_{1,t_n}, \ldots, {\Delta}_{p, t_n}]^{\intercal}$}, respectively, are defined in order to ensure flexible, yet shock--robust, point forecasts at each $t_{s^*}=t_{n+1}, \ldots, t_{n+n^*}$. More specifically,  \smash{$\hat{\bbeta}$} corresponds to the mean of the smoothing distribution for $\bbeta_{t_n}$, whereas  \smash{$\hat{\bDelta}$} is defined as the  median of the estimates of $\partial \beta_j(t)/\partial t$, over the last 25 years $t_n, \ldots, t_{n-24}$ computed under the smoothing distribution in \eqref{eq13}, for every $j=1, \ldots, p$. Since the smoothing distribution in \eqref{eq13} is a Gaussian, these estimates coincide with the elements having position $2+3(j-1)$ in $\bmu_{t_{s}|t_{1:n}}$, for each $j=1, \ldots, p$ and $t_s=t_n, \ldots, t_{n-24}$. Rather than projecting forward in time a single global dynamic component, as in  \citet{lee1992modeling}, strategy \eqref{eq15}--\eqref{eq16} gains accuracy by extrapolating  multiple time dynamics corresponding to the different age classes, while relying on a natural initialization for the drift terms which  leverages  the ability of the proposed model to explicitly learn dynamics also in the first order derivatives quantifying rates of change in mortality levels. To ensure robustness to shocks while adapting to the most recent globally--linear trend dynamics, the starting drift term $\hat{\bDelta}$ is set to the median, rather than the mean, of these estimates over the last 25 years, which was found to be a robust default time horizon in the application to multiple countries in Section~\ref{sec_4}. This is  in line with similar results in  \citet{lee1992modeling} on the time horizon to condition on.

Although the above strategy yields improved point forecasts, the derivation of effective predictive intervals requires an accurate estimate of  \smash{$\sigma^2_\omega$} and \smash{${\sigma}^2_{\Delta}$} in \eqref{eq16}, along with the variance parameter \smash{$\sigma^2_{\psi}$} in the observation equation which yields the forecasted \smash{$\hat{\bef}_{t_{s^*}}$} in \eqref{eq15}, namely
 \begin{eqnarray}
\log \bm_{t_{s^*}} =[\bg_1, \ldots, \bg_p] {\bbeta}_{t_{s^*}} + \bnu_{t_{s^*}}, \qquad \bnu_{t_{s^*}}\sim \text{N}_k({\bf 0},  \sigma^2_\psi {\bf I}_k).
\label{eq16a}
\end{eqnarray}
Consistent with the strategy adopted for deriving the point forecasts, these three variances are obtained via maximum marginal likelihood under model \eqref{eq16}--\eqref{eq16a} applied to data from $t_{n-24}$  until $t_{n}$.  To suitably connect model \eqref{eq9}--\eqref{eq10} with the simpler formulation in  \eqref{eq16}--\eqref{eq16a}, ${\bDelta}_{t_{n-24}}$ is initialized at \smash{$\mbox{N}_p(\hat{\bmu}_{\bDelta}, \mbox{diag}(\hat{\bsigma}^2_{\Delta}))$}, where $\hat{\bmu}_{\bDelta}$ is the  median of the estimated $\partial \beta_j(t)/\partial t$, for each $j=1, \ldots, p$, under the smoothing distribution in \eqref{eq13}, over the 25 years preceding  $t_{n-24}$, while the generic entry $\hat{\sigma}_{\Delta_j}^2$  in \smash{$\hat{\bsigma}^2_{\Delta}$} corresponds to the sample variance of the median estimate \smash{$\hat{\mu}_{\Delta_j}$}, for $j=1, \ldots, p$, computed via a set of simulations from the smoothing distribution. The initial vector  ${\bbeta}_{t_{n-24}}$ is instead assumed to follow the $p$--variate Gaussian distribution where the mean vector is obtained from the one--step--ahed projection under \eqref{eq16} of the smoothing estimate for ${\bbeta}_{t_{n-25}}$  provided by model \eqref{eq9}--\eqref{eq10}, whereas the covariance  matrix coincides with that of the predictive distribution at $t_{n-25}$ under model \eqref{eq9}--\eqref{eq10}.

From a practical perspective, the above strategies can be still implemented via standard \texttt{R} libraries for time series analysis, and, as previously discussed, are reminiscent of the two--step approach  by  \citet{lee1992modeling}, which relies on an in--sample estimate of the global time--specific effect and then fits, for future projections, a random walk plus drift model on the subset of these time effects corresponding to a suitably defined most recent window.

%%%%%%%%%%%%%%%%%%%%%%%%%%%%%%%%%%%%%%%%%%%%%%
%%%%%%%%%%%%%%%%%%%%%%%%%%%%%%%%%%%%%%%%%%%%%%
\section{Learning and Forecasting of Mortality Across Countries}\label{sec_4}
In order to quantify the improvements provided by the novel \textsc{bsp}  developed  in Sections~\ref{sec_2}--\ref{sec_3}, we consider extensive analyses and performance comparisons with a main focus on the gender--specific age--period log--mortality rates for the countries discussed in Section~\ref{sec_1.1},  across a wide time horizon that spans from $1933$ until either $2019$ or $2020$ depending on data availability in the \texttt{Human} \texttt{Mortality} \texttt{Database} at \url{https://www.mortality.org/}. 

In modeling the age--period log--mortality rates for the four countries that are object of our study we follow common practice in demographic studies \citep[e.g.,][]{haberman2011comparative,currie2016fitting,camarda2019smooth} and consider a separate analysis for each combination of gender--country, leading to a total of eight different implementations of the  \textsc{bsp} model in Sections~\ref{sec_2}--\ref{sec_3} and of selected state--of--the--art competitors. More specifically, for each combination of gender and country we study the age--period log--mortality rates via the  Gaussian state--space approximation in \eqref{eq9}--\eqref{eq10} of the  \textsc{bsp} model defined in Section~\ref{sec_21}, employing the $p=20$ \textsc{b}--spline bases $g_1(x), \ldots, g_{20}(x)$ illustrated in Figure~\ref{fig:spline}, and considering Mat\'ern covariance functions for $\mathcal{K}(\bar{x}_j,\bar{x}_l; \bgamma_\beta)$ and $\mathcal{K}(\bar{x}_j,\bar{x}_l; \bgamma_a)$ in \eqref{eq8} \citep[e.g.,][Ch.\ 4]{williams2006gaussian}. The specification of a total of $p=20$ bases over the observed age range $x_1=0, \ldots, x_{101}=100$ is motivated by a similar choice for the age dimension in the bivariate  \textsc{b}--splines construction of \citet{camarda2019smooth}. Consistent with the graphical evidence in Figure~\ref{fig:data}, these  \textsc{b}--splines are more dense at early and late ages to achieve increased flexibility in capturing local dynamic variations  for such classes. For the same reasons, the first  \textsc{b}--spline $g_1(x)$ is the only one that is active at $x_1=0$ since mortality at age $0$ is known to display peculiar patterns relative to those from age $1$ onward, thereby requiring increased flexibility relative to the other classes  \citep{camarda2019smooth}. The Mat\'ern covariance functions parameters $\bgamma_\beta$ are instead set at $(0.5,1)$ to induce local borrowing of information only across close ages, whereas no correlation is enforced on the instantaneous means to increase the flexibility in modeling  shocks affecting only specific age classes. Although these covariance parameters could be estimated, together with $\sigma_{\overline{m}}^2$, $\sigma_{\beta}^2$, \smash{$\sigma_{a}^2$} and $\lambda$, via maximum marginal likelihood under model \eqref{eq9}--\eqref{eq10}, as clarified in Table~\ref{tab} and in Figures~\ref{fig:smoothing}--\ref{fig:covid}, the suggested settings provide robust default choices to accurately learn and forecast several mortality patterns across different countries. In fact, moderate changes in the Mat\'ern covariance parameters as well as in the number and location of the \textsc{b}--spline bases  did not change the final conclusions. Notice also that, although the squared exponential covariance function provides another routinely--implemented alternative, such a function can be recovered as a special case of the Mat\'ern one which, therefore, yields a more general class with a higher degree of flexibility. 

\begin{figure}[t]
	\centering
		\includegraphics[width=1\textwidth,height=0.5\textwidth]{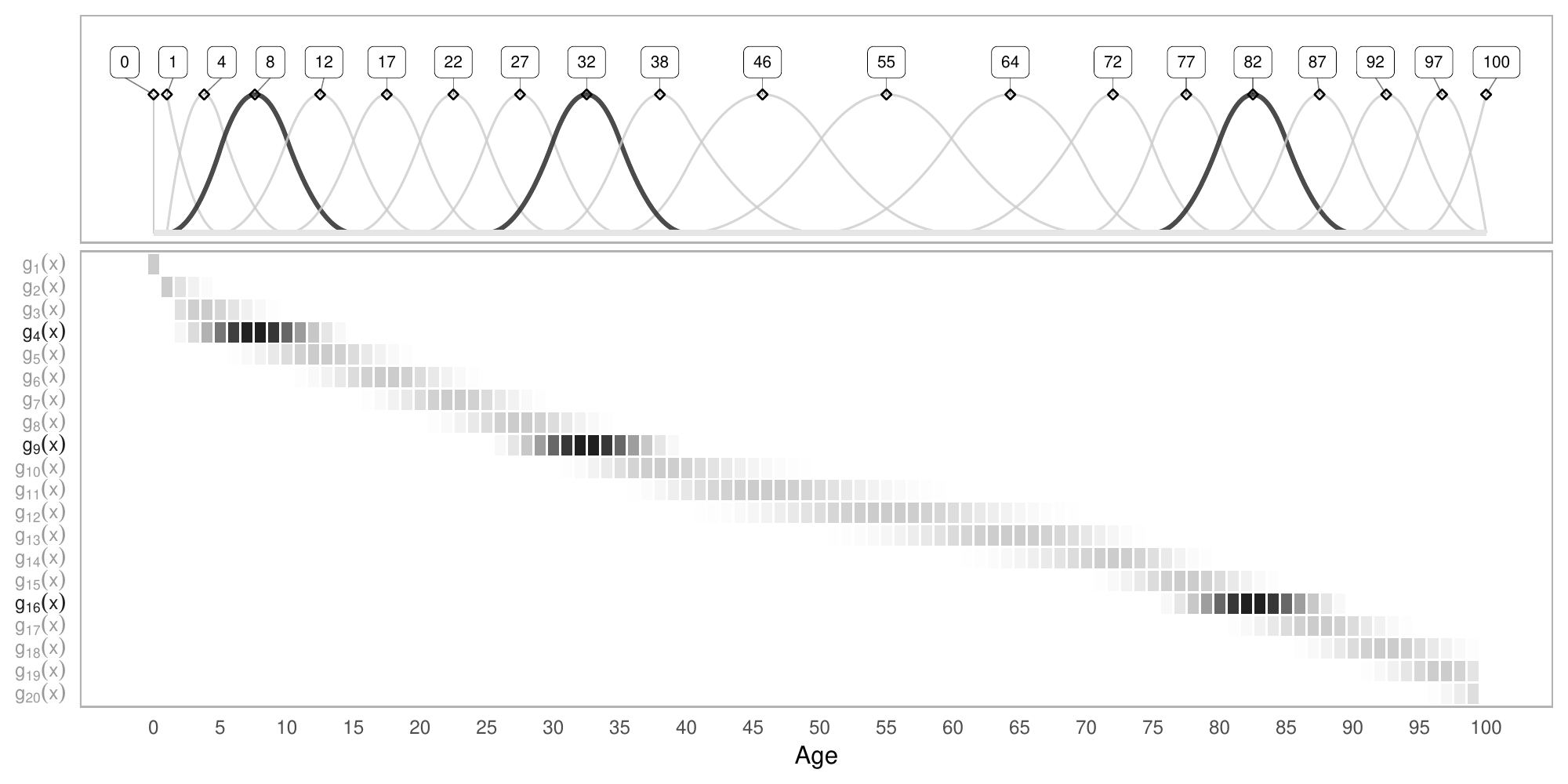}
		\caption{ \footnotesize Graphical representation of the $p=20$ selected \textsc{b}--spline bases $g_1(x), \ldots, g_{20}(x)$ along with a heatmap  clarifying the intensity of each spline in the corresponding age range. The number $\bar{x}_j$ associated to each spline $g_j(x)$ in the top panel denotes the age at which such a spline takes its maximum value, for each $j=1, \ldots, 20$. For illustrative purposes, three  bases and the corresponding active range of ages are highlighted in black. }
	\label{fig:spline}
\end{figure}

As a first assessment, we evaluate in Table~\ref{tab} the performance in point forecasting of the \textsc{bsp} formulation proposed in Sections~\ref{sec_2}--\ref{sec_3}, and quantify the improvements relative to the state--of--the--art competitors discussed in Section~\ref{sec_1}. These include the classical \citet{lee1992modeling}  (\textsc{lc}) and age--period--cohort (\textsc{apc}) \citep{holford1983estimation,osmond1985using} models,  along with the subsequent developments and extensions in  \citet{renshaw2006cohort} (\textsc{rh}),  \citet{cairns2006two} (\textsc{cbd}), \citet{hyndman2007robust} (\textsc{hu}), \citet{plat2009stochastic} (\textsc{plat}) and  \citet{camarda2019smooth} (\textsc{cp}), which currently represent the leading strategies within mortality forecasting. As illustrated in the code available at \url{https://github.com/fpavone/BSP-mortality}, these models can be readily implemented via standard \texttt{R} functions in the packages \texttt{StMoMo} \citep{millossovich2018stmomo},  \texttt{demography} \citep{hyndman2006demography} and in the \texttt{R} code within  the supplementary materials of  \citet{camarda2019smooth}. Recalling Section~\ref{sec_3}, parameter estimation, inference and forecasting under the Gaussian state--space formulation of the proposed \textsc{bsp} approach can be effectively implemented via the \texttt{R} package   \texttt{KFAS} \citep{helske2017kfas}. 

\begin{table}[t]
\renewcommand{\arraystretch}{1}
\centering
\caption{\footnotesize For the eight methods under analysis and ten predictive horizons, overall median of the absolute difference between the forecasted and observed log--mortality rates computed from all the country--gender--age--year combinations. Bold values denote the best performance for each  predictive horizon, whereas the gray column corresponds to the proposed \textsc{b}--spline process with locally--adaptive dynamic coefficients.  For the three top performing methods, the first and third quartiles of the absolute differences are also reported within brackets.  } 
\small
\scalebox{0.92}{
\begin{tabular}{lrrrrrrrr}
  \hline
{\footnotesize Steps  ahead} \quad & \cellcolor{gray!20}  \textsc{bsp} & \textsc{cp}  &  \textsc{hu}  &   \textsc{plat}  &  \textsc{rh}  &  \textsc{apc}  &  \textsc{lc} &  \textsc{cbd}  \\ 
  \hline
   $1$ \quad & \cellcolor{gray!20} ${\bf 0.032}$ {\footnotesize \bf $[0.01, 0.07]$}  & $0.033$  {\footnotesize $[0.01, 0.07]$}  & $0.038$ {\footnotesize $[0.02, 0.09]$}    & $0.097$ & $0.068$ & $0.126$ & $0.107$ & $0.142$ \\ 
     $2$ \quad  &\cellcolor{gray!20}   ${\bf 0.037}$ {\footnotesize \bf $[0.02, 0.08]$}  & $0.040$ {\footnotesize $[0.02, 0.08]$} & $0.047$ {\footnotesize $[0.02, 0.09]$}& $0.104$ & $0.077$ & $0.138$ & $0.114$ & $0.152$ \\ 
     $3$ \quad & \cellcolor{gray!20}  ${\bf  0.044}$ {\footnotesize \bf $[0.02, 0.09]$}  & $0.048$ {\footnotesize $[0.02, 0.10]$}  & $0.056$ {\footnotesize $[0.03, 0.11]$}& $0.113$ & $0.087$ & $0.150$ & $0.122$ & $0.163$ \\ 
    $4$ \quad & \cellcolor{gray!20}  ${\bf 0.050}$ {\footnotesize \bf $[0.02, 0.10]$}  & $0.057$ {\footnotesize $[0.03, 0.11]$}& $0.064$ {\footnotesize $[0.03, 0.12]$} & $0.122$ & $0.095$ & $0.162$ & $0.129$ & $0.174$ \\ 
     $5$ \quad & \cellcolor{gray!20}  ${\bf  0.056}$ {\footnotesize \bf $[0.03, 0.11]$}   & $0.066$ {\footnotesize $[0.03, 0.12]$}& $0.072$ {\footnotesize $[0.03, 0.14]$}  & $0.131$ & $0.105$ & $0.176$ & $0.137$ & $0.187$ \\ 
     $6$ \quad& \cellcolor{gray!20}  ${\bf 0.063}$ {\footnotesize \bf $[0.03, 0.12]$}   & $0.073$ {\footnotesize  $[0.03, 0.13]$}& $0.080$ {\footnotesize $[0.04, 0.15]$} & $0.140$ & $0.117$ & $0.190$ & $0.144$ & $0.200$ \\ 
     $7$ \quad  & \cellcolor{gray!20}  ${\bf 0.070}$ {\footnotesize \bf $[0.03, 0.13]$}   & $0.081$ {\footnotesize $[0.04, 0.15]$}& $0.088$ {\footnotesize $[0.04, 0.16]$}& $0.149$ & $0.129$ & $0.205$ & $0.154$ & $0.214$ \\ 
     $8$ \quad & \cellcolor{gray!20}  ${\bf 0.076}$ {\footnotesize \bf $[0.03, 0.14]$}   & $0.088$ {\footnotesize $[0.04, 0.16]$}& $0.094$ {\footnotesize $[0.04, 0.18]$}& $0.159$ & $0.143$ & $0.222$ & $0.161$ & $0.229$ \\ 
     $9$ \quad & \cellcolor{gray!20}  ${\bf 0.083}$ {\footnotesize \bf $[0.04, 0.16]$}    & $0.095$ {\footnotesize $[0.04, 0.17]$} & $0.102$ {\footnotesize $[0.05, 0.19]$} & $0.169$ & $0.156$ & $0.233$ & $0.171$ & $0.246$ \\ 
    $10$ \quad & \cellcolor{gray!20}  ${\bf 0.093}$ {\footnotesize \bf $[0.04, 0.17]$}   & $0.105$ {\footnotesize $[0.05, 0.19]$} & $0.110$ {\footnotesize $[0.05, 0.21]$}& $0.179$ & $0.174$ & $0.251$ & $0.178$ & $0.263$ \\ 
   \hline
\end{tabular}}
\label{tab}
\end{table}

Table~\ref{tab} summarizes the performance in point forecasting  of the above strategies over different time horizons, ranging from 1--step--ahead to 10--step--ahead. For the eight gender--country combinations, these forecasts are obtained by sequentially fitting each model on the observed age--period log--mortality rates from $1933$ up until a last year ranging from $1990$ to $2010$, and then predicting, for each of these final years from $1990$ to $2010$, the age--period log--mortality rates in the subsequent ten years. This produces, under each model and step--ahead predictive horizon, a total $ 2 \times 101 \times 21 $ forecasts per country --- except for Italy whose data are available only until 2019 --- which correspond to the different combinations of gender, ages and last observation time, thereby providing  a large sample of predictions to accurately compare the different methods. Table~\ref{tab} displays the overall median of the absolute differences between these forecasts and the actual observed log–mortality rates, across countries, gender, ages and the last observation times. Results provide empirical evidence for the improved forecasting accuracy of the proposed \textsc{bsp}, which outperforms all the state--of--the--art alternatives for every predictive horizon. As expected, \textsc{cp} \citep{camarda2019smooth}  and \textsc{hu}  \citep{hyndman2007robust} are the most competitive alternatives. Recalling Section~\ref{sec_1}, also these strategies rely on a flexible basis expansion, but are not as effective as the proposed \textsc{bsp} in incorporating all the core structures of age--period mortality surfaces, thereby allowing our procedure to further improve forecasting accuracy. Notice that, although the magnitude of these improvements is not always remarkable, the  \textsc{bsp} approach remains systematically more accurate than all the methods considered, both in terms of medians and the two quartiles. Such a finding was further confirmed in additional studies for other countries in the \texttt{Human} \texttt{Mortality} \texttt{Database} (i.e., Czech Republic, Denmark and France) and when comparing the forecasting performance of the different methods via the mean squared error, rather than the median of the absolute error. The latter measure is preferred in Table~\ref{tab} since it provides a more robust and direct measure of the actual distance between the forecasted and observed mortality rates. We shall also emphasize that, in these contexts, even a small reduction in the predictive errors for the log--mortality rates can have a major impact in population forecasts since, as is clear from equation \eqref{eq1}, such rates are multiplied by the central exposed to risk $\textsc{e}_{xt}$ when modeling the total death counts $d_{xt}$, with $\textsc{e}_{xt}$ in the order of tens--to--hundreds of thousands in common population analyses at the country level. This reasoning applies also to other demographic measures derived as a function of the  mortality rates, such as, for example, the life expectancy at birth, whose forecasts can be directly obtained via the \texttt{R} package \texttt{demography} from those produced for $m_{xt}$. Also in this case,   \textsc{bsp} was still found to outperform  all the state--of--the--art competitors and almost halved the 10--step--ahead predictive errors of both \textsc{cp} and \textsc{hu}.

From a computational perspective, all the methods analyzed in Table~\ref{tab}, including the proposed \textsc{bsp} strategy, facilitate tractable and scalable implementations which yield runtimes for estimation, inference and forecasting always below one minute. This is several orders of magnitude lower than the yearly time scale at which mortality data are typically analyzed, thereby providing effective solutions for rapid updating of inferences and forecasts.

To further clarify the major advantages of the proposed  \textsc{bsp} construction, we also considered predictive comparisons against direct implementations of the simpler building--blocks underlying the proposed formulation and forecasting approach. More specifically,  instead of relying on the strategy outlined in Section~\ref{sec_32}, we considered forecasts obtained either under the predictive distribution in \eqref{eq12} of the original \textsc{bsp} formulation in \eqref{eq9}--\eqref{eq10}, or from the direct implementation of separate nested Gaussian processes for the trajectories $f_t(x)$ of every age $x \in \mathcal{X}$ rather than employing the more structured formulation proposed within equations \eqref{eq2}--\eqref{eq4}. Focusing again on the time horizon ranging from 1--step--ahead to 10--step--ahead forecasts,  the overall medians of the absolute differences between the forecasted and observed log--mortality rates were $[0.035, 0.045, 0.057, 0.067,  0.079, 0.091, 0.104, 0.116, 0.129, 0.145]$ for the first alternative strategy and $[0.058, 0.109, 0.179, 0.264, 0.370, 0.490, 0.630, 0.785, 0.959, 1.149]$ for the second. Comparing these results with those in the first column of Table~\ref{tab} clarifies the key advantages of the proposed \textsc{bsp} construction that achieves an improved predictive accuracy by carefully borrowing information across ages via a structured \textsc{b}--spline representation with dependence across the dynamic coefficients, which is subsequently leveraged to develop the parsimonious, yet effective, forecasting strategy outlined in Section~\ref{sec_32}.

\begin{table}[t]
\renewcommand{\arraystretch}{1}
\centering
\caption{\footnotesize For the eight methods under analysis and ten predictive horizons, relative proportion, computed  from all the country--gender--age--year combinations, of the $95\%$ predictive intervals containing the observed log--mortality rates. Bold values denote the best performance for each  predictive horizon, whereas the gray column corresponds to the proposed \textsc{b}--spline process with locally--adaptive dynamic coefficients.  } 
\small
\scalebox{0.92}{
\begin{tabular}{l@{\hspace{5em}}r@{\hspace{2em}}r@{\hspace{2em}}r@{\hspace{2em}}r@{\hspace{2em}}r@{\hspace{2em}}r@{\hspace{2em}}r@{\hspace{2em}}r}
  \hline
{\footnotesize Steps  ahead} \quad & \cellcolor{gray!20}  \textsc{bsp} & \textsc{cp}  &  \textsc{hu}  &   \textsc{plat}  &  \textsc{rh}  &  \textsc{apc}  &  \textsc{lc} &  \textsc{cbd}  \\ 
  \hline
 1 & \cellcolor{gray!20} ${\bf 0.955}$ & $0.514$ & $0.980$ & $0.404$ & $0.394$ & $0.299$ & $0.224$ & $0.325$ \\ 
     2 & \cellcolor{gray!20} ${\bf 0.954}$ & $0.481$ & $0.982$ &  $0.519$ & $0.450$ & $0.378$ & $0.290$ & $0.388$ \\ 
     3 &\cellcolor{gray!20}  ${\bf 0.953}$ & $0.454$ & $0.982$ &  $0.571$ & $0.462$ & $0.407$ & $0.336$ & $0.429$ \\ 
     4 & \cellcolor{gray!20} ${\bf 0.952}$ & $0.428$ & $0.984$ &  $0.603$ & $0.480$ & $0.430$ & $0.367$ & $0.459$ \\ 
     5 & \cellcolor{gray!20} ${\bf 0.950}$ & $0.403$ & $0.984$ &  $0.625$ & $0.482$ & $0.446$ & $0.385$ & $0.479$ \\ 
     6 & \cellcolor{gray!20} ${\bf 0.948}$ & $0.370$ & $0.982$ &  $0.641$ & $0.478$ & $0.456$ & $0.397$ & $0.494$ \\ 
     7 &\cellcolor{gray!20}  ${\bf 0.948}$ & $0.343$ & $0.979$ &  $0.642$ & $0.476$ & $0.455$ & $0.414$ & $0.504$ \\ 
     8 & \cellcolor{gray!20} ${\bf 0.946}$ & $0.324$ & $0.977$ &  $0.645$ & $0.464$ & $0.456$ & $0.422$ & $0.508$ \\ 
     9 & \cellcolor{gray!20} ${\bf 0.945}$ & $0.311$ & $0.976$ &  $0.646$ & $0.449$ & $0.460$ & $0.426$ & $0.514$ \\ 
    10 & \cellcolor{gray!20} ${\bf 0.937}$ & $0.294$ & $0.972$ & $0.646$ & $0.439$ & $0.451$ & $0.429$ & $0.512$ \\ 
   \hline
\end{tabular}}
\label{tab2}
\end{table}

Let us conclude the analysis of  forecasting performance by assessing the calibration of  predictive intervals under the different methods considered in Table~\ref{tab}. To this end, Table~\ref{tab2} displays the relative proportion, computed from all the different combinations of country--gender--age--year, of the $95\%$ predictive intervals which contain the observed log--mortality rates. Also in this setting, the \textsc{bsp} intervals computed under the methods illustrated in Section~\ref{sec_32} achieve improved overall calibration relative to those obtained under the competing strategies. Comparing the results in Table~\ref{tab2} with those  in Table~\ref{tab}, the poor performance of \textsc{plat}, \textsc{rh},  \textsc{apc},  \textsc{lc} and  \textsc{cbd} is  mainly attributable to the bias in the point forecasts at which such intervals are centered, whereas \textsc{cp} suffers from an underestimation of the predictive variance, possibly due to challenges in the implementation of the employed bootstrap strategy within the time--series context. As illustrated in Table~\ref{tab2}, \textsc{hu} is the only competitive strategy, although it exhibits an over--coverage tendency with intervals having a similar length to those obtained under the proposed  \textsc{bsp} construction. We shall  emphasize that, when stratifying by age,  calibration of  \textsc{bsp} intervals is generally less accurate at younger ages than older ones, thus motivating additional future refinements.

\begin{figure}
	\centering
		\includegraphics[width=0.99\textwidth]{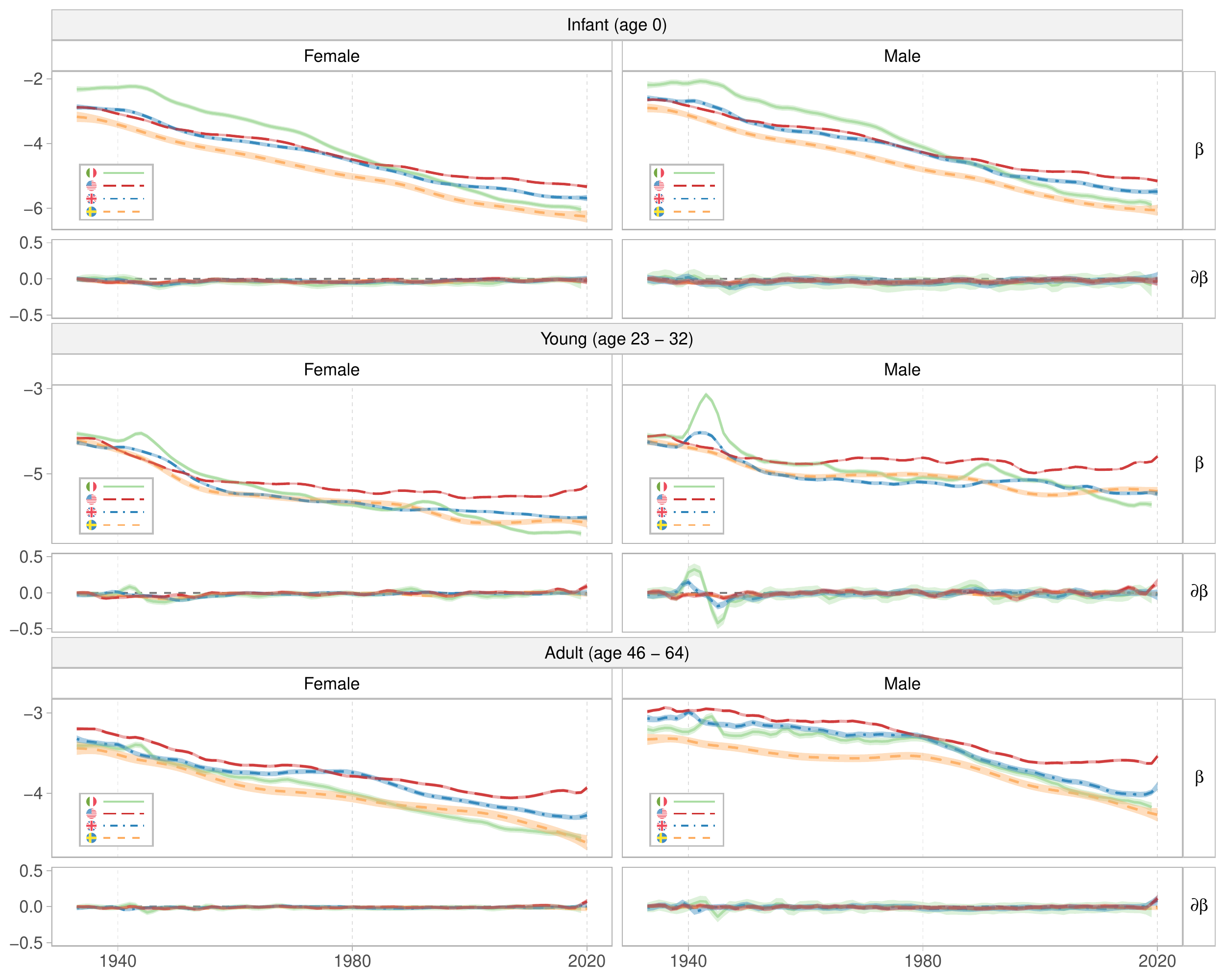}
		\caption{\footnotesize For females (left) and males (right), and the four countries under analysis, time trajectories of three representative \textsc{b}--splines coefficients along with the corresponding first derivatives, as obtained from the smoothing distribution under the proposed \textsc{bsp} model. The lines correspond to the trajectories of the means, whereas the shaded areas denote the pointwise $95\%$ credible intervals. See the online article for a color version of this figure.}	
	\label{fig:smoothing}
\end{figure}

The improvements in forecasting performance of \textsc{bsp} motivate additional analyses and country comparisons of the age--period mortality surfaces, which are further facilitated by the interpretable construction of the proposed model  in Sections~\ref{sec_2}--\ref{sec_3}. In fact, as illustrated in Figures~\ref{fig:smoothing}--\ref{fig:covid}, \textsc{bsp} allows to formally  study and compare changes in the mortality patterns across years and specific ages via inference on the location and variability of the smoothing distribution for the coefficients of the splines active in those age classes. These selected trajectories are displayed in Figure~\ref{fig:smoothing}, along with the corresponding first derivatives, and highlight interesting differences across countries in the dynamic evolution of age--specific mortality rates. For instance, in the first row of Figure~\ref{fig:smoothing},  \textsc{bsp} learns a peculiar trajectory for infant mortality in Italy, characterized by a structural break soon after the World War {\rm II} which leads to  a faster decay in infant mortality with respect to other countries. This remarkable change is aligned with the so--called {\em Italian miracle},  a phase of  rapid economic growth and improved life conditions after the World War II \citep[e.g.,][]{ginsborg1990history}, progressively bringing infant mortality in Italy to even lower levels than those registered in countries such as the United Kingdom and the United States. The local adaptivity of \textsc{bsp} can be instead appreciated in the second row of  Figure~\ref{fig:smoothing}, where the proposed model learns the rapid mortality increment corresponding to the World War II, which, as expected, is mainly evident for Italian males, albeit visibile also for males in the United Kingdom. Notice that, since most of the United States military deaths happened abroad, these counts do not contribute to \textsc{us} mortality as recorded in the \texttt{Human} \texttt{Mortality} \texttt{Database}. In the second row of  Figure~\ref{fig:smoothing}, \textsc{bsp} also learns an excess--mortality peak in Italy for both young males and females in the late  '80s and early '90s. This  provides quantitative evidence for the rapid and severe combined effect of \textsc{aids}, car accidents and overdoses in Italy during those years for young age classes \citep[see e.g.,][]{conti1994aids,conti1997premature}. Despite this shock, Italy displays a generally decreasing trend in the spline coefficients associated to young age classes, that interestingly departs from the general stagnation, or even the increasing trend, which \textsc{bsp} learns for the other three countries, especially in the last two decades. This is particularly remarkable in the United States, which display peculiar mortality patterns characterized by slower mortality decrements or even increments, mainly evident since the '80s, for all the three age classes analyzed in Figure~\ref{fig:smoothing}. These quantitive findings  further support a number of studies on the recent  \textsc{us} mortality crisis as a consequence of specific disparities and vulnerabilities associated with young and adult age classes \citep[][]{ho2010us,woolf2019life,glei2021us,preston2021excess,case2021deaths}. As is clear from Figure~\ref{fig:smoothing}, all the smoothing distributions analyzed are characterized by limited uncertainty, thus supporting the reliability of these findings.

\begin{figure}[t]
	\centering
		\includegraphics[width=0.98\textwidth]{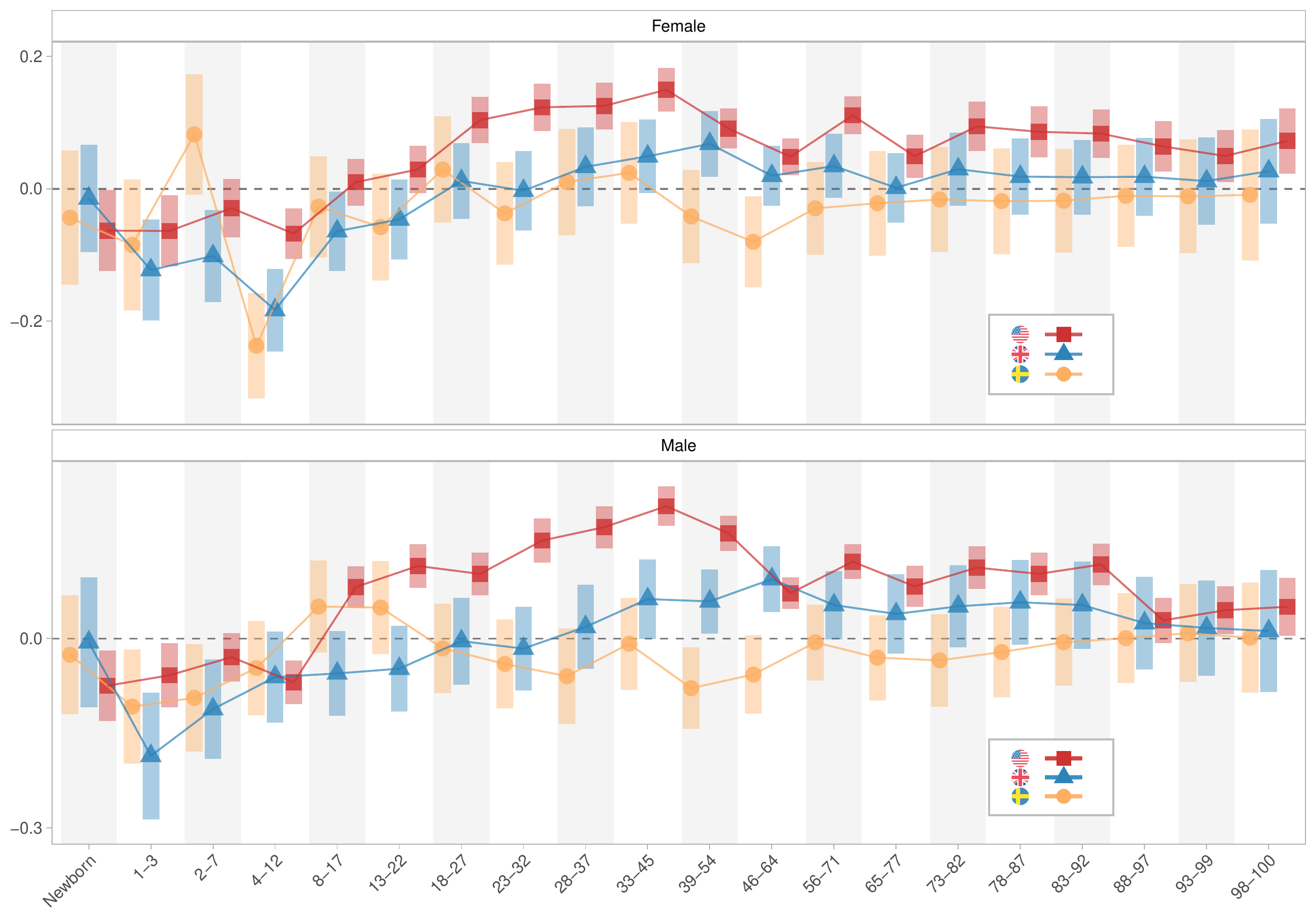}
		\caption{\footnotesize For both females and males, means (colored symbols) and $95\%$ credible intervals (colored boxes) of the smoothing distribution for the difference between the spline coefficients in year 2020 and the corresponding average over 2014--2019, for  the United States, United Kingdom and Sweden. Data for Italy in 2020 are not yet available in the \texttt{Human} \texttt{Mortality} \texttt{Database}. This representation provides a summary of excess mortality in 2020. See the online article for a color version of this figure.}	
	\label{fig:covid}
\end{figure}

Notably, the above patterns are also associated with differences in the rates of change of mortality  levels during \textsc{covid}--19, as inferred from the analysis of the first derivatives of  the three splines coefficients in Figure~\ref{fig:smoothing} for year 2020. The ability of  \textsc{bsp} to explicitly model and quantify uncertainty also in these rates of change in mortality trends crucially allows to learn a mortality shock during \textsc{covid}--19 in young age classes only for the United States and not for the other countries under analysis; see the panels $\partial \beta$ in the second row of Figure~\ref{fig:smoothing}. These findings are further expanded in Figure~\ref{fig:covid} where the focus is on the smoothing distribution of the differences between each spline coefficient in year 2020 and its average in the previous five years, for Sweden, United Kingdom and United States; data for Italy in year 2020 are not yet available in the \texttt{Human} \texttt{Mortality} \texttt{Database} at \url{https://www.mortality.org/}. Consistent with the discussion of Figure~\ref{fig:smoothing},  \textsc{bsp} infers a noticeably--high excess mortality in the United States, for both females and males, which is surprisingly visibile from very young age classes onward, and whose magnitude is much higher than in the United Kingdom and Sweden. This fundamental finding further corroborates recent studies on the association between \textsc{covid}--19 effects and the peculiar pre--existing \textsc{us} disparities and vulnerabilities, especially in relation to risk factors \citep[e.g.,][]{wiemers2020disparities}. Despite the less stringent  policies adopted in Sweden, the \textsc{covid}--19 mortality shock for such a country is less evident than the one registered in the United Kingdom and United States. It is important to emphasize that also Sweden experienced an excess mortality during the first and second wave of the \textsc{covid}--19 pandemic \citep[e.g.,][]{juul2022mortality}. However, when aggregating all--causes mortality at a yearly scale, such increments become less visibile and systematic, pointing toward a possible mortality displacement effect \citep{juul2022mortality}, also known as {\em harvesting}; i.e., a phase of excess deaths followed by a mortality deficit that has a balancing effect when aggregating at a larger time scale.

%%%%%%%%%%%%%%%%%%%%%%%%%%%%%%%%%%%%%%%%%%%%%%
%%%%%%%%%%%%%%%%%%%%%%%%%%%%%%%%%%%%%%%%%%%%%%
\section{Conclusion and Future Research Directions}\label{sec_5}
This article proposes a novel  \textsc{b}--spline process with locally--adaptive dynamic coefficients for accurate learning and forecasting of mortality patterns across ages and periods. Such a process decomposes the age--period mortality surface as a flexible, yet interpretable, function of age, and crucially treats the dynamics of this function across periods via a suitable stochastic process of time that explicitly incorporates the core structures of mortality evolution through a set of stochastic differential equations. This allows to (i) incorporate and learn differences in the time patterns of mortality across age classes, while borrowing information between close ages, (ii) explicitly infer and project not only age--specific mortality trends, but also the corresponding rates of change, (iii) characterize dynamics that fluctuate among periods of rapid and slow variation, (iv) devise simple and accurate forecasting strategies for log--mortality rates which are both flexible and shock--robust, (v) develop computationally--efficient methods for filtering, smoothing and prediction of mortality patterns via closed--form Kalman filter recursions.

To the best of our knowledge, none of the solutions currently available in the literature accounts for all the aforementioned properties within a single formulation. In fact, as illustrated in the application in Section~\ref{sec_4}, the proposed model generally improves forecasting performance and crucially expands the set of findings which can be obtained from the analysis of age--period mortality data. This perspective can open new avenues to formally compare differences in mortality patterns across ages, countries and years, while quantifying  possible heterogeneities in the rate of change of mortality and in the impact of shocks, such as, for example, the recent \textsc{covid}--19 pandemic,  for which we infer notable differences across countries.

Besides providing an important contribution to the literature on mortality modeling, the proposed formulation also motivates several future advancements.  A relevant direction is to extend the \textsc{b}--spline process within Section~\ref{sec_2} for joint modeling of multiple populations, possibly from high, middle and low income countries. Although the \texttt{Human} \texttt{Mortality} \texttt{Database} has data only for the first group, such an extension can be accomplished  within our formulation by considering a mixture of \textsc{b}--spline processes that would further allow to cluster countries characterized by similar age--period mortality patterns. This facilitates borrowing of information for countries with low population size or studies at a local level, and incorporates improved coherence in mortality forecasts, an important aspect in recent multi--population studies \citep[e.g.,][]{li2005coherent,wen2021fitting,wang2022multi}. Alternatively, it would be of interest to specify country--specific   \textsc{b}--spline processes with locally--adaptive dynamic coefficients and then induce dependence among such processes via a suitable graphical model \citep{lauritzen1996graphical}, thus allowing inference on conditional independence structures in age--period mortality dynamics among different countries, while borrowing information to improve inference and forecasting. This can be particularly useful also in the joint modeling of male and female mortality.

The above directions are also of interest when the focus is on joint modeling of mortality patterns for different causes--of--death, rather than countries \citep[e.g.,][]{kjaergaard2019forecasting}. 
 
 %%%%%%%%%%%%%%%%%%%%%%%%%%%%%%%%%%%%%%%%%%%%%%%%%%%%%%%%%%%%%
 
 \begin{acks}[Acknowledgments]
 We are grateful to Stefano Mazzuco and Emanuele Aliverti for the~insightful discussion on the data analyzed.  This research was supported by the \textsc{miur}--\textsc{prin} 2017 project \textsc{Select} (20177BRJXS). In addition, Daniele Durante acknowledges support by the   \textsc{mur}--\textsc{prin} 2022 project \textsc{Caronte} (2022KBTEBN, funded by the European Union -- Next Generation EU), during the final revision of the present article.
 \end{acks}

%%%%%%%%%%%%%%%%%%%%%%%%%%%%%%%%%%%%%%%%%%%%%%%%%%%%%%%%%%%%%
%%%%%%%%%%%%%%%%%%%%%%%%%%%%%%%%%%%%%%%%%%%%%%%%%%%%%%%%%%%%%
 \appendix

\section{Proofs of Propositions}\label{app2}

\begin{proof}[Proof of Proposition~\ref{p1}]
The proof of Proposition~\ref{p1} adapts similar derivations considered by \citet{liang2014modeling} in the context of binomial logistic--normal distributions. In particular, under model \eqref{eq1} we have that  \smash{$(d_{xt} \mid \overline{m}_{xt}) \stackrel{\mbox{\scriptsize ind}}{\sim} \mbox{Poisson}(\textsc{e}_{xt}\overline{m}_{xt})$}. Hence, $\E(m_{xt} \mid \overline{m}_{xt})=\E(d_{xt}/\textsc{e}_{xt} \mid \overline{m}_{xt})= \overline{m}_{xt}$ and $\mbox{var}(m_{xt} \mid \overline{m}_{xt})=\mbox{var}(d_{xt}/\textsc{e}_{xt} \mid \overline{m}_{xt})= \overline{m}_{xt}/\textsc{e}_{xt}$. Therefore,  for fixed $\overline{m}_{xt}$, it follows that $\mbox{var}(m_{xt} \mid \overline{m}_{xt}) \to 0$ as $\textsc{e}_{xt} \to \infty$, which implies mean square convergence of $m_{xt}$ to $\overline{m}_{xt}$ as $\textsc{e}_{xt} \to \infty$. As a consequence, we have also convergence in probability, which further guarantees  convergence in distribution, i.e., $\lim_{\textsc{e}_{xt}  \to \infty}\mbox{pr}(d_{xt}/\textsc{e}_{xt} \leq u \mid \overline{m}_{xt})=\mathbbm{1}(\overline{m}_{xt} \leq u)$. Leveraging this result and applying the dominated convergence theorem, it follows that 
\begin{eqnarray*}
\begin{split}
\lim_{\textsc{e}_{xt}  \to \infty}\mbox{pr}(d_{xt}/\textsc{e}_{xt}\leq u)&=\lim_{\textsc{e}_{xt}  \to \infty}\int_{0}^{\infty}\mbox{pr}(d_{xt}/\textsc{e}_{xt}\leq u \mid \overline{m}_{xt}) p(\overline{m}_{xt}) \partial \overline{m}_{xt}\\
&=\int_{0}^{\infty}\lim_{\textsc{e}_{xt}  \to \infty}\mbox{pr}(d_{xt}/\textsc{e}_{xt} \leq u \mid \overline{m}_{xt}) p(\overline{m}_{xt}) \partial \overline{m}_{xt}\\
&=\int_{0}^{\infty}\mathbbm{1}(\overline{m}_{xt} \leq u) p(\overline{m}_{xt}) \partial \overline{m}_{xt}=\mbox{pr}(\overline{m}_{xt} \leq u).
\end{split}
\end{eqnarray*}
Hence, $d_{xt}/\textsc{e}_{xt}$ converges in distribution to the log--normal assumed for $\overline{m}_{xt}$ in equation \eqref{eq1} and, as a direct consequence of the continuous mapping theorem, it follows that $\log m_{xt}=\log (d_{xt}/\textsc{e}_{xt}) \to \mbox{\normalfont N}(f_{t}(x), \sigma_{{\overline{m}}}^2)$ in distribution, as $\textsc{e}_{xt} \to~\infty$, for any $x \in \mathcal{X}$ and $t\in \mathcal{T}$.
\end{proof}

\begin{proof}[Proof of Proposition~\ref{p2}]
The proof of Proposition~\ref{p2} follows directly by applying the derivations in Appendix A.6 of \citet{zhu2013locally} to the stochastic differential equations defined in~\eqref{eq3}--\eqref{eq4}.
\end{proof}

%%%%%%%%%%%%%%%%%%%%%%%%%%%%%%%%%%%%%%%%%%%%%%%%%%%%%%%%%%%%%

\end{document}